\documentclass[11pt,letterpaper]{article}

\pdfoutput=1

\def\showauthornotes{1}
\usepackage{amsmath}
\usepackage{amssymb}
\usepackage{amsthm}
\usepackage[in]{fullpage}
\usepackage{forloop}
\usepackage{bm}
\usepackage{bbm}
\usepackage{xspace}
\usepackage[l2tabu, orthodox]{nag}
\usepackage{nth}
\usepackage{color}
\usepackage{array}
\usepackage{paralist}
\usepackage{caption}
\usepackage{flafter}
\usepackage{tikz}
\usepackage{xifthen}
\usepackage{url}
\usepackage{times}

\usetikzlibrary{calc}
\usetikzlibrary{matrix}
\usetikzlibrary{shapes,arrows}
\usetikzlibrary{decorations.pathmorphing}
\usetikzlibrary{decorations.pathreplacing}
\usetikzlibrary{intersections}
\usetikzlibrary{positioning}

\newtheorem{theorem}{Theorem}
\newtheorem{observation}[theorem]{Observation}

\newtheorem*{fact*}{Fact}

\numberwithin{theorem}{section}

\newtheorem{thm}[theorem]{Theorem}
\newtheorem{lemma}[theorem]{Lemma}

\newtheorem{obs}[theorem]{Observation}

\theoremstyle{definition}

\makeatletter
\newtheorem*{rep@theorem}{\rep@title}
\newcommand{\newreptheorem}[2]{
\newenvironment{rep#1}[1]{
 \def\rep@title{#2 \ref{##1}}
 \begin{rep@theorem}}
 {\end{rep@theorem}}}

\makeatother

\newreptheorem{reduction}{Reduction}
\newreptheorem{lemma}{Lemma}

\newcommand{\oddconst}{\ensuremath{\mathsf{odd}}}
\newcommand{\evenconst}{\ensuremath{\mathsf{even}}}
\newcommand{\unconst}{\ensuremath{\mathsf{unconstrained}}}

\colorlet{darkgreen}{green!50!black}

\newcommand{\defcal}[1]{\expandafter\newcommand\csname c#1\endcsname{{\mathcal{#1}}}}
\newcounter{ct}
\forLoop{1}{26}{ct}{
    \edef\letter{\Alph{ct}}
    \expandafter\defcal\letter
}

\ifnum\showauthornotes=1
\newcommand{\Authornote}[2]{{\sffamily\small\color{red}{[#1: #2]}}}
\newcommand{\Authornotecolored}[3]{{\sffamily\small\color{#1}{[#2: #3]}}}
\newcommand{\Authorcomment}[2]{{\sffamily\small\color{gray}{[#1: #2]}}}
\newcommand{\Authorstartcomment}[1]{\sffamily\small\color{gray}[#1: }

\newcommand{\Authorfnote}[2]{\footnote{\color{red}{#1: #2}}}
\newcommand{\Authorfixme}[1]{\Authornote{#1}{\textbf{??}}}
\newcommand{\Authormarginmark}[1]{\marginpar{\textcolor{red}{\fbox{\Large #1:!}}}}
\else
\newcommand{\Authornote}[2]{}
\newcommand{\Authornotecolored}[3]{}
\newcommand{\Authorcomment}[2]{}
\newcommand{\Authorstartcomment}[1]{}

\newcommand{\Authorfnote}[2]{}
\newcommand{\Authorfixme}[1]{}
\newcommand{\Authormarginmark}[1]{}
\fi

\newcommand{\Sinv}{\ensuremath{S_{\mathsf{inv}}}\xspace}
\newcommand{\SI}{\ensuremath{S_{\mathcal{I}}}\xspace}
\newcommand{\sI}{\ensuremath{\sigma_{\mathcal{I}}}\xspace}
\newcommand{\SO}{\ensuremath{S_{\mathcal{O}}}\xspace}
\newcommand{\sO}{\ensuremath{\sigma_{\mathcal{O}}}\xspace}
\newcommand{\Obar}{\ensuremath{\overline{O}}\xspace}
\newcommand{\dJ}{\ensuremath{\mathsf{deg}_J}\xspace}

\newcommand{\hide}[1]{}
\newcommand{\sgraph}[1]{\ensuremath{\mathcal{SG}
\ifthenelse{\equal{#1}{}}{}{(#1)}
}}
\newcommand{\cgraph}[1]{\ensuremath{\mathcal{CG}
\ifthenelse{\equal{#1}{}}{}{(#1)}
}}
\newcommand{\cpath}[2]{\ensuremath{P_{#1}
\ifthenelse{\equal{#2}{}}{}{(#2)}
}}

\DeclareMathOperator*{\argmin}{arg\,min}

\newcommand{\UncSol}{\mathcal{I}}

\definecolor{colorempty}{RGB}{255,255,255}
\definecolor{colorodd}{RGB}{255,31,0}
\definecolor{coloreven}{RGB}{0,0,191}
\definecolor{colorsolopen}{RGB}{0,0,0}
\definecolor{coloroptopen}{RGB}{127,127,127}
\definecolor{coloradjtoz}{RGB}{0,0,0}
\definecolor{colorya}{RGB}{0,127,0}
\definecolor{coloryb}{RGB}{0,0,127}
\definecolor{coloryc}{RGB}{255,127,127}

\newcommand{\drawEmpty}[3]{
    \draw [draw=colorempty, fill=colorempty] (#1*#3, #2*#3) -- (#1*#3+#3, #2*#3) -- (#1*#3+#3, #2*#3+#3) -- (#1*#3, #2*#3+#3) -- cycle;
}

\newcommand{\drawOdd}[4]{
    \draw [draw=colorodd, dash pattern=on \pgflinewidth off 1.5pt, line width=#4] (#1*#3, #2*#3) -- (#1*#3+#3, #2*#3) -- (#1*#3+#3, #2*#3+#3) -- (#1*#3, #2*#3+#3) -- cycle;
}

\newcommand{\drawEven}[4]{
    \draw [draw=coloreven, line width=#4] (#1*#3, #2*#3) -- (#1*#3+#3, #2*#3) -- (#1*#3+#3, #2*#3+#3) -- (#1*#3, #2*#3+#3) -- cycle;
}

\newcommand{\drawSolOpen}[3]{
    \draw [draw=colorsolopen, fill=colorsolopen] (#1*#3+#3, #2*#3) -- (#1*#3+#3, #2*#3+#3) -- (#1*#3, #2*#3+#3) -- cycle;
}

\newcommand{\drawOptOpen}[3]{
    \draw [draw=coloroptopen, fill=coloroptopen] (#1*#3+#3, #2*#3) -- (#1*#3, #2*#3) -- (#1*#3, #2*#3+#3) -- cycle;
}

\newcommand{\drawOddOO}[4]{
    \drawOdd{#1}{#2}{#3}{#4}
    \drawEmpty{#1}{#2}{#3}
    \drawSolOpen{#1}{#2}{#3}
    \drawOptOpen{#1}{#2}{#3}
}

\newcommand{\drawEvenOO}[4]{
    \drawEven{#1}{#2}{#3}{#4}
    \drawEmpty{#1}{#2}{#3}
    \drawSolOpen{#1}{#2}{#3}
    \drawOptOpen{#1}{#2}{#3}
}

\newcommand{\drawOddCO}[4]{
    \drawOdd{#1}{#2}{#3}{#4}
    \drawEmpty{#1}{#2}{#3}
    \drawOptOpen{#1}{#2}{#3}
}

\newcommand{\drawEvenCO}[4]{
    \drawEven{#1}{#2}{#3}{#4}
    \drawEmpty{#1}{#2}{#3}
    \drawOptOpen{#1}{#2}{#3}
}

\newcommand{\drawOddOC}[4]{
    \drawOdd{#1}{#2}{#3}{#4}
    \drawEmpty{#1}{#2}{#3}
    \drawSolOpen{#1}{#2}{#3}
}

\newcommand{\drawEvenOC}[4]{
    \drawEven{#1}{#2}{#3}{#4}
    \drawEmpty{#1}{#2}{#3}
    \drawSolOpen{#1}{#2}{#3}
}

\newcommand{\drawClient}[3]{
    \draw[draw=black, fill=white, thick] (#1*#3+#3/2, #2*#3+#3/2) circle (#3/2);
}

\newcommand{\drawSpecial}[3]{
    \draw[draw=black, fill=white, thick] (#1*#3+#3/2, #2*#3-#3/4) -- (#1*#3-#3/4, #2*#3+#3/2) -- (#1*#3+#3/2, #2*#3+#3+#3/4) -- (#1*#3+#3+#3/4, #2*#3+#3/2) -- cycle;
    \draw (#1*#3+#3/2, #2*#3+#3/2) node[align=center] {$z$};
}

\newcommand{\drawTextA}[3]{
    \draw[font=\bfseries] (#1*#3+#3/2, #2*#3+#3/2) node[align=center] {(A)};
}

\newcommand{\drawTextB}[3]{
    \draw[font=\bfseries] (#1*#3+#3/2, #2*#3+#3/2) node[align=center] {(B)};
}

\newcommand{\drawSolAssign}[5]{
    \draw[->, draw=colorsolopen, ultra thick] (#1*#5+#5/2, #2*#5+#5/2) -- (#3*#5+#5/2, #4*#5+#5/2);
}

\newcommand{\drawOptAssign}[5]{
    \draw[->, draw=coloroptopen, ultra thick] (#1*#5+#5/2, #2*#5+#5/2) -- (#3*#5+#5/2, #4*#5+#5/2);
}

\newcommand{\drawOptAssignCurve}[7]{
    \draw[draw=coloroptopen, ultra thick] (#1*#7+#7/2, #2*#7+#7/2) .. controls (#5*#7+#7/2, #6*#7+#7/2) .. (#3*#7+#7/2, #4*#7+#7/2);
}

\newcommand{\drawEdgeGrey}[5]{
    \draw[->, draw=coloradjtoz, thin] (#1*#5+#5/2, #2*#5+#5/2) -- (#3*#5+#5/2, #4*#5+#5/2);
}

\newcommand{\drawEdgeGreyCurve}[7]{
    \draw[draw=coloradjtoz, thin] (#1*#7+#7/2, #2*#7+#7/2) .. controls (#5*#7+#7/2, #6*#7+#7/2) .. (#3*#7+#7/2, #4*#7+#7/2);
}

\newcommand{\drawYa}[5]{
    \draw[->, densely dotted, draw=colorya, ultra thick] (#1*#5+#5/2, #2*#5+#5/2) -- (#3*#5+#5/2, #4*#5+#5/2);
}

\newcommand{\drawYb}[5]{
    \draw[->, loosely dotted, draw=coloryb, ultra thick] (#1*#5+#5/2, #2*#5+#5/2) -- (#3*#5+#5/2, #4*#5+#5/2);
}

\newcommand{\drawYc}[5]{
    \draw[->, draw=coloryc, ultra thick] (#1*#5+#5/2, #2*#5+#5/2) -- (#3*#5+#5/2, #4*#5+#5/2);
}

\begin{document}
\title{{Constant-Factor Approximation Algorithms\\for Parity-Constrained Facility Location Problems\footnote{
This work was supported by the National Research Foundation of Korea (NRF) grant funded by the Korea government (MSIT) (No. NRF-2019R1C1C1008934).
This work was supported by the National Research Foundation of Korea (NRF) grant funded by the Korea government (MSIT) (No. NRF-2016R1C1B1012910).
This research was supported by the Yonsei University Research Fund of 2018-22-0093.
}}}

\author{
Kangsan Kim\thanks{Department of Computer Science, Yonsei University, South Korea.
Email: \texttt{wongrikera@yonsei.ac.kr}}
,
Yongho Shin\thanks{Department of Computer Science, Yonsei University, South Korea.
Email: \texttt{yshin@yonsei.ac.kr}}
,
Hyung-Chan An\thanks{Corresponding author. Department of Computer Science, Yonsei University, South Korea. \hspace{12em}\mbox{ }
\mbox{ }\mbox{ }\mbox{ }\mbox{ }\mbox{ }\mbox{ }\mbox{ }Email: \texttt{hyung-chan.an@yonsei.ac.kr}}
}

\date{} 

\maketitle

\setcounter{page}{0}
\maketitle
\thispagestyle{empty}

\begin{abstract}
\emph{Facility location} is a prominent optimization problem that has inspired a large quantity of both theoretical and practical studies in combinatorial optimization. Although the problem has been investigated under various settings reflecting typical structures within the optimization problems of practical interest, little is known on how the problem behaves in conjunction with parity constraints. This shortfall of understanding was rather disturbing when we consider the central role of \emph{parity} in the field of combinatorics.
In this paper, we present the first constant-factor approximation algorithm for the facility location problem with parity constraints. We are given as the input a metric on a set of \emph{facilities} and \emph{clients}, the opening cost of each facility, and the \emph{parity requirement}---\oddconst, \evenconst, or \unconst---of every facility in this problem. The objective is to open a subset of facilities and assign every client to an open facility so as to minimize the sum of the total opening costs and the assignment distances, but subject to the condition that the number of clients assigned to each open facility must have the same parity as its requirement.

Although the unconstrained facility location problem as a relaxation for this parity-constrained generalization has unbounded gap, we demonstrate that it yields a structured solution whose parity violation can be corrected at small cost. This correction is prescribed by a $T$-join on an auxiliary graph constructed by the algorithm. This graph does not satisfy the triangle inequality, but we show that a carefully chosen set of shortcutting operations leads to a cheap and \emph{sparse} $T$-join. Finally, we bound the correction cost by exhibiting a combinatorial multi-step construction of an upper bound. At the end of this paper, we also present the first constant-factor approximation algorithm for the parity-constrained $k$-center problem, the bottleneck optimization variant.
\end{abstract}

\medskip
\noindent
{\small \textbf{Keywords:}
facility location problems, approximation algorithms, clustering problems, parity constraints, \linebreak $k$-center problem
}

\newpage

\section{Introduction}

\emph{Parity} plays a central role in a myriad of topics in combinatorics. This is so natural that one would not need examples; yet, we remark that, as a short sample of previous works, Schrijver and Seymour~\cite{schrijver1994} for example studied the packing of odd paths, Everett et al.~\cite{everett1997} and Kami\'nski \& Nishimura~\cite{kaminski2012} considered induced path parities in connection with the theory of perfect graphs, and Kakimura et al.~\cite{kakimura2012} studied the packing of parity-constrained cycles intersecting a given vertex set.
Naturally, there also exists a large volume of previous research that incorporates parity constraints into different combinatorial optimization problems. Submodular function minimization~\cite{grotschel1981ellipsoid,grotschel1984}, the minimum cut problem~\cite{padberg1982,benczur2000}, the shortest path problem (\emph{cf.}~\cite{grotschel1981weakly}), and the connected subgraph problem~\cite{sebo2013,cheriyan2015} are all examples of such problems. However, introducing parity constraints to combinatorial optimization problems usually results in a significant level of added complexity in their algorithms, and perhaps due to this difficulty, not all parity-constrained combinatorial optimization problems are as well studied as one would expect from the centrality of parity in this field.

The \emph{facility location problem} is one of the prominent optimization problems that has guided a large volume of studies in both computer science and operations research (see, e.g.,~\cite{balinski1964,alfred1963,alan1964,john1963}). In this problem, we are given as the input a set of \emph{facilities} and a set of \emph{clients}, along with the opening cost of each facility and the metric distance between every pair of facility and client. The goal of the problem is to choose a subset of facilities to \emph{open} and a clustering that \emph{assigns} every client to an open facility, so as to minimize the sum of the facility opening costs and the distance between each client and the facility it is assigned to. While the facility location problem also served as a test bed on which a variety of algorithmic theories were developed, another primary reason the problem has attracted the interests of many researchers is that it closely reflects the structure of optimization problems witnessed in practice. Precisely for this reason, facility location problems are studied in a wide variety of settings that better reflect typical constraints imposed on the problem,
including the capacitated version~\cite{alfred1963,korupolu2000,pal2001,bansal2012,an2017} that places an upper bound on the number of clients assigned to an open facility,
online and/or dynamic variants~\cite{cygan2018, goranci2018, eisenstat2014evolving, lammersen2008},
mobile facility location~\cite{ahmadian2013},
planar versions~\cite{marx2015},
and the lower-bounded version that imposes a lower bound on the number of clients assigned to an open facility~\cite{li2019}.
Unfortunately, in conjunction with parity constraints, it was not previously known how this well-understood problem behaves on the other hand.

This paper aims at filling this gap. In the \emph{$O$-facility location problem}, a subset of facilities $O$ is specified as part of the input, in addition to the usual input for the unconstrained facility location problem. The goal of the problem is still to find a minimum-cost subset of open facilities with a clustering of the clients, but now we also need to ensure an additional constraint that the number of clients assigned to each open facility $i$ must be odd if $i\in O$, and even otherwise. It is particularly surprising that we do not have a proper understanding of this generalized version to this date, especially when we consider its practical relevance. In many problems that seek an optimal clustering, we sometimes have a strong preference for either parity of the cluster sizes. For example, Ahamad and Ammar~\cite{ahmed1989} demonstrate that the performance of a \emph{distributed database system} (DDBS), measured by success rates and mean response times, depends on the parity of the number of storage sites. In fact, this preferred parity is determined as a function of the server failure rates and the ratio between the number of read and write transactions. Thus, the task of clustering a given set of storage sites into multiple instances of distributed database systems, which host different applications whose parameters vary, can be formulated as an $O$-facility location problem. Preference on a particular parity can also be witnessed in other distributed system design settings (see, e.g., \cite{zookeeper}) or even outside the realm of computer science and operations research~\cite{frank1971,menon2011}.

Another extensively studied clustering problem is the \emph{$k$-center problem}, a bottleneck optimization variant of facility location. Given an integer $k$ and a metric on a set of nodes, the goal of this problem is to choose at most $k$ nodes as \emph{centers} and assign every node to one of these centers. In contrast to the (min-sum) facility location problem, the objective here is to minimize the \emph{maximum} assignment distance. In applications, depending on the nature of the assignment distances, one may be more interested in this bottleneck objective rather than min-sum: imagine for example that in the problem of clustering storage sites the assignment distances were given as the propagation delays of communication links, in which case the maximum propagation delay would have greater significance compared to the ``total'' delay.
The $k$-center problem was studied in various settings as well, including capacitated versions~\cite{khuller2000,cygan2012}, dynamic/online settings~\cite{lang2018,chan2018}, planar graphs~\cite{eisenstat2014planar,demaine2005}, and lower-bounded versions~\cite{ene2013}.

However, a similar lack of understanding again exists for parity constraints. In the \emph{parity-constrained $k$-center problem}, each node is labeled with one of the following as part of the input: \oddconst, \evenconst, or \unconst. The objective of the problem is now to find an assignment where the number of nodes assigned to each center has the same parity as its label. Note that, unlike its min-sum counterpart, we allowed \emph{three} types of parity constraints (or absence thereof). This is because, for the facility location problem, there is a very simple equivalence argument (see Section~\ref{sec:dproofs}) between the $O$-facility location problem and the version that allows unconstrained facilities.

In this paper, we present the first $O(1)$-approximation algorithm for the parity-constrained facility location problem: first for the special case where $O=\emptyset$, and then for the general case.

\begin{theorem}\label{thm:alleven}
There exists an $O(1)$-approximation algorithm for the $\emptyset$-facility location problem.
\end{theorem}

\begin{theorem}\label{thm:general}
There exists an $O(1)$-approximation algorithm for the $O$-facility location problem for arbitrary $O$.
\end{theorem}

\noindent
We also present the first $O(1)$-approximation algorithm for the parity-constrained $k$-center problem.

\begin{theorem}\label{thm:kcentermain}
There exists an $O(1)$-approximation algorithm for the parity-constrained $k$-center problem.
\end{theorem}

The difficulty of the classic unconstrained facility location problem lies in the fact that it is a ``joint optimization'' problem: it is trivial to find an optimal assignment when the set of open facilities is given, but the simultaneous optimization of the choice of open facilities along with the assignment makes the problem difficult and, in fact, NP-hard. The $O$-facility location problem is a generalization of the unconstrained facility location problem and therefore inherits this difficulty. Moreover, even when the set of open facilities is given, it is not as easy to find an optimal assignment for this problem (although polynomial-time solvable).

In order to obtain good approximation algorithms for the unconstrained facility location problem, many algorithmic tools have been used. In particular, linear programming (LP) relaxations and methods based on them have been  successful~\cite{shmoys1997,jain2001,jain2002,byrka2007,li2011}.
Unfortunately, however, the $O$-facility location problem does not appear amenable to LP-based techniques, and it is easy to show that the standard LP relaxation devised in the context of the unconstrained problem has an unbounded integrality gap, i.e., the LP optimum can be away from the true optimum by an arbitrarily large factor. In fact, even the integral optimum to the unconstrained instance obtained by dropping the parity constraints can be arbitrarily away from the true optimum.\footnote{Consider an instance with two pairs of an even-constrained facility and a client, where the distance within each pair is zero and one across. Both opening costs are zeroes.} Despite this gap, to our surprise, we will prove that the approximation algorithms for the unconstrained facility location problem can serve as a useful subroutine of an approximation algorithm for the parity-constrained generalization.

In Section~\ref{sec:alleven} and Appendix~\ref{app:even}, we present our $O(1)$-approximation algorithm for the \emph{all-even} case, i.e., $O=\emptyset$. The algorithm begins with finding a minimum perfect matching on the set of clients. Using the fact that every facility is assigned an even number of clients in an optimal solution, we can ``shortcut'' the optimal solution into a perfect matching, bounding the minimum cost of a perfect matching. We then reduce the given instance to an instance of the unconstrained facility location problem by designating one of the two matched clients as the representative, at the cost of a constant multiplicative factor in the approximation ratio.

This clean approach, however, crucially relies on the fact that every facility is even-constrained, and does not extend to the general case. This necessitates a totally different approach, which is presented in Section~\ref{sec:general}. The first step of our $O(1)$-approximation algorithm for the general case is to drop the parity constraints and solve this instance using an algorithm for the unconstrained problem.\footnote{The analysis in Section~\ref{sec:general} treats the algorithm for the unconstrained problem as a black box to give the ratio of 6.464; using a bi-factor approximation algorithm leads to an improvement over this ratio.} The cost of the obtained approximate solution is a lower bound on the true optimum, but as was noted earlier, it may be arbitrarily smaller than the true optimum. However, we show that the parity violation of this \emph{initial solution} can be repaired by performing a set of three types of operations: reassigning one client from an open facility to another, opening a new facility, and closing down an open facility after reassigning all its clients to another facility. We observe that, interestingly, it suffices to permit the third type of operation only when the facility is odd-constrained, which in turn allows such a set of operations to be encoded as a \emph{sparse} $T$-join on an auxiliary graph. The auxiliary graph has three types of edges corresponding to the three types of operations. The key step of the analysis is to show that the minimum cost of a $T$-join of the auxiliary graph is bounded by a linear combination of the initial solution (despite the gap) and an optimal solution.

In contrast to the $O$-facility location problem, the parity-constrained $k$-center problem turns out to be quite amenable to existing techniques. In Appendix~\ref{app:kcenter}, we show that the standard method to solve bottleneck optimization problems, namely guessing the optimum and considering an unweighted graph of ``admissible'' edges along which assignments can be made, remains useful for this problem. The algorithm considers each connected component of this graph to perform a long ``chain'' of reassignments, similar to other variants of the $k$-center problem~\cite{cygan2012,an2015}. 
\footnote{We obtain a $6$-approximation algorithm that runs in $O(|V|^2\log |V|)$ time; whilst this approximation ratio is obtained under the standard problem definition where there is no distinction between facilities and clients, we can easily extend the analysis to obtain an $O(1)$-approximation algorithm even when the distinction exists.}

\paragraph{Future directions.}
One of the interesting questions that follow this paper is whether we can write an algorithmically useful LP relaxation for the $O$-facility location problem. As we could use an LP-based approximation algorithm for the unconstrained facility location problem as a subroutine of our algorithm, one could say that our algorithm can technically be an LP-based algorithm; this, of course, is not a satisfactory answer. Rather than having to solve an LP relaxation which itself is parameterized by a rounded integral solution to another relaxation (which is the case for our algorithm), it would be interesting to have a single relaxation that can be separated in polynomial time and solved to obtain an $O(1)$-approximate lower bound on the optimum. Recall that, for the minimum-cost $T$-join problem, an exact and polynomial-time separable relaxation exists~\cite{edmonds1973}.

\vspace{.7ex}

Another intriguing future direction is in introducing parity constraints to further combinatorial optimization problems. As we noted earlier, to our surprise considering the prominence of parity in combinatorics, there remain many parity-constrained optimization problems yet to be studied. We envision that a further understanding of our algorithm, particularly if we can positively answer our first open question, may lead to extending our knowledge to other parity-constrained optimization problems.

\vspace{2ex}

\section{Preliminaries}
\label{sec:probdefn}

\paragraph{Problem definition.}
As the input of the \emph{$O$-facility location problem}, we are given a set of \emph{facilities} $F$, a set of \emph{clients} $D$, \emph{opening costs} $f:F\to\mathbb{Q}_{\geq0}$, \emph{assignment costs}\footnote{Assignment costs are sometimes defined only between facilities and clients. In this ``bipartite'' case, the domain of $c$ will be defined as $F\times D$ instead. These two definitions, however, are equivalent, since we can deduce inter-facility (and inter-client) distances by computing the metric closure of the given ``bipartite'' assignment cost.} $c:\binom{F\cup D}{2}\to\mathbb{Q}_{\geq 0}$ satisfying the triangle inequality, and a set of facilities $O\subseteq F$ that, if open, are required to be assigned odd number of clients.

A feasible solution to the problem is given by a set of \emph{open} facilities $S\subseteq F$ and an \emph{assignment} of clients $\sigma:D\to S$ to the open facilities. In order for $(S,\sigma)$ to be a feasible solution, it must satisfy the parity constraints: for all $i\in S\cap O$, $|\sigma^{-1}(i)|$ must be odd; for all $i\in S\cap \overline{O}$, $|\sigma^{-1}(i)|$ must be even. The objective is to find a feasible solution that minimizes the total solution cost, defined as $\sum_{i\in S}f(i)+\sum_{j\in D}c(\sigma(j),j)$.

\paragraph{An equivalent problem definition.}
Alternatively, we can define our problem as taking a parity constraint function $\pi:F\to\{\mathsf{odd},\mathsf{even},\mathsf{unconstrained}\}$ instead of $O$. In this case, the parity constraint is redefined as follows: for each $i\in F$, $|\sigma^{-1}(i)|$ must be odd if $\pi(i)=\mathsf{odd}$ and even if $\pi(i)=\mathsf{even}$. (If $\pi(i)=\mathsf{unconstrained}$, we do not impose any parity constraints on $|\sigma^{-1}(i)|$.) 

\begin{obs}\label{obs:equiv}
The two problem definitions are equivalent.
\end{obs}

\noindent Its proof is deferred to Section~\ref{sec:dproofs}. Note that this observation shows the NP-hardness of the $O$-facility location problem.

\paragraph{Notation.}
To simplify the presentation, we introduce the following shorthands: for any $S\subseteq F$, let $f(S):=\sum_{i\in S}f(i)$, and for any $\sigma:D\to S$, let $c(\sigma):=\sum_{j\in D}c(\sigma(j),j)$. Using this notation, the objective function of the problem can be rewritten as $f(S)+c(\sigma)$. For $D^\prime\subseteq D$ and $\sigma:D\to S$, let $\sigma|_{D^\prime}:D^\prime \to S$ denote the restriction of $\sigma$ to $D^\prime$ as the new domain, i.e., $\sigma|_{D^\prime}(j)=\sigma(j)$ for all $j\in D^\prime$. Accordingly, $c(\sigma|_{D^\prime})$ is defined as $c(\sigma|_{D^\prime}):=\sum_{j\in D^\prime}c(\sigma|_{D^\prime}(j),j)=\sum_{j\in D^\prime}c(\sigma(j),j)$. Let $\sigma^{-1}:S\to D$ denote the inverse function of $\sigma$, i.e., $\sigma^{-1}(i):=\{j\in D\mid \sigma(j)=i\}$. We will slightly abuse the notation by letting $\sigma^{-1}(i):=\emptyset$ for $i\in F\setminus S$.

\paragraph*{Additional definitions.}

Let $G=(V,E)$ be a graph. For $T\subseteq V$, we say $J\subseteq E$ is a \emph{$T$-join} if, for every vertex $v\in V$, the number of edges in $J$ that are incident with $v$ is odd if and only if $v\in T$. Given a weighted graph, the minimum-cost $T$-join can be found in polynomial time~\cite{edmonds1973} (see also \cite{schrijver2003}).

Given $T\subseteq V$, we say $U\subseteq V$ is \emph{$T$-odd} if $|U \cap T|$ is odd, and $Y\subseteq E$ is a \emph{$T$-join dominator} if, for every $T$-odd set $U \subseteq V$, there exists at least one edge in $Y$ that has exactly one endpoint in $U$. 
\begin{lemma}[\cite{edmonds1973,schrijver2003}]
Given a weighted graph $G=(V,E)$ with $T\subseteq V$ and a $T$-join dominator $Y$, the minimum cost of a $T$-join on $G$ is no greater than the cost of $Y$.
\end{lemma}

Given two sets $P$ and $Q$, let $P \triangle Q$ denote the symmetric difference of the sets, i.e., $P \triangle Q := (P \setminus Q) \cup (Q \setminus P)$.
Finally, let $\mathcal{A}_\mathsf{FL}$ denote a $\rho_\mathsf{FL}$-approximation algorithm for the \emph{unconstrained} facility location problem.

\section{\texorpdfstring{$O(1)$}{Lg}-Approximation for the All-Even Case}\label{sec:alleven}

In this section, we present a constant-factor approximation algorithm for a special case of the problem where the parity constraint of every facility is even, i.e., $O=\emptyset$.

\paragraph{Our algorithm.}
We first find a minimum-cost perfect matching $M^\star$ on $D$, using $c$ as the cost function.
For each $e \in M^\star$, we independently choose one of the two endpoints of $e$ uniformly at random. Let $j_e$ be the chosen client and $\widehat{j_e}$ be the remaining one.
Let $D^\prime$ be the set of chosen clients, i.e., $ D^\prime = \{j_e \;|\; e \in M^\star\}$.
We now construct an unconstrained facility location instance where the client set is replaced with $D^\prime$. The rest of the input ($F$, $c$, and $f$) remains the same.
We execute $\mathcal{A}_\mathsf{FL}$ on this instance; let $S_\UncSol\subseteq F$ and $\sigma_\UncSol:D^\prime\to S_\UncSol$ denote the solution returned by $\mathcal{A}_\mathsf{FL}$. We construct a solution $S_\mathsf{ALG}$ and $\sigma_\mathsf{ALG}:D\to S_\mathsf{ALG}$ to our problem as follows: we choose $S_\mathsf{ALG}$ simply as $S_\UncSol$.
For each remaining client $\widehat{j_e} \in D \setminus D^\prime$, we assign $\widehat{j_e}$ to the same facility to which its pair is assigned, i.e., 
\begin{equation*}
    \sigma_\mathsf{ALG}(j) = \left\{
        \begin{array}{ll}
            \sigma_\UncSol (j), & \text{if } j \in D^\prime, \\
            \sigma_\UncSol (j_e), & \text{if } j = \widehat{j_e} \text{ for some } e \in M^\star.
        \end{array}
    \right.
\end{equation*}

This algorithm is a randomized $2\rho_\mathsf{FL}$-approximation algorithm for the problem. We present the full analysis in Appendix~\ref{app:even}.

\section{General Case}
\label{sec:general}

In this section, we present an $O(1)$-approximation algorithm for the general $O$-facility location problem.

\subsection{Our algorithm}

\paragraph{Outline.}
We start with a brief outline of our algorithm. As the first step of the algorithm, we execute $\mathcal{A}_\mathsf{FL}$ for the unconstrained facility location problem on the given input, but with its parity constraints dropped.
Let $S_{\mathcal{I}}\subseteq F$ and $\sigma_{\mathcal{I}}:D\to S_{\mathcal{I}}$ be the algorithm's output. Note that $(S_{\mathcal{I}}, \sigma_{\mathcal{I}})$ may be infeasible since we dropped the parity constraints; though, our algorithm will use it as the ``initial'' solution and correct the parities at small cost.

The second step of our algorithm is to construct an auxiliary weighted graph $G$ and a set of vertices $T\subseteq V(G)$. The construction is designed so that a $T$-join (almost) prescribes a way to correct the parities. Naturally, our algorithm will find a minimum-cost $T$-join.

Then the last step of our algorithm is to modify the initial solution as indicated by the minimum-cost $T$-join on the auxiliary graph. We first post-process the minimum-cost $T$-join we found to obtain a \emph{sparse} $T$-join. We will show that this ``sparsified'' $T$-join specifies a modification to the initial solution that restores the parity constraints.

In what follows, we describe the last two steps in more detail.

\paragraph{Construction of the auxiliary graph.}
We say a facility $i\in\SI$ is \emph{invalid} if its parity constraint is violated in the initial solution.
Let $\Sinv$ denote the set of invalid facilities, i.e., $\Sinv:=\{i\in O\mid |\sI^{-1}(i)|\textrm{ is even}\}\cup \{i\in \Obar\mid |\sI^{-1}(i)|\textrm{ is odd}\}$.

The vertex set of the auxiliary graph $G$ is $F\cup \{z\}$ for an artificial vertex $z\notin F$. Let $E$ be the edge set of the auxiliary graph and $\gamma:E\to\mathbb{Q}_{\geq 0}$ be the edge cost.
The following are three types of edges that we create in $G$.\begin{itemize}\setlength{\itemsep}{0pt}
\item (\emph{reassign edges}) For each pair of distinct facilities $i,i'\in F$,
we create an edge $(i,i')$ in the auxiliary graph with cost $\gamma(i,i'):=c(i,i')$.
\item (\emph{opening edges}) For each odd-constrained, initially closed facility $i \in O \setminus\SI$, we create an edge $(z,i)$ with cost $\gamma(z,i):=f(i)$.
\item (\emph{closing edges}) This last type of edges is created only if $|\SI|\geq 2$ or $|\overline{O} \setminus\SI|\geq 1$. For each odd-constrained, initially open facility $i \in O \cap S_\mathcal{I}$, we create an edge $(z,i)$ with cost
\begin{equation}\label{e:closingcost}
\gamma(z, i) := \min\begin{cases}
    \displaystyle \min_{i^\prime \in S_\mathcal{I} \setminus \{i\}}
    \left[ \left| \sigma^{-1}_{\mathcal{I}}(i) \right| \cdot c(i, i^\prime) \right]&\\[1em]
    \displaystyle \min_{i^\prime \in {\overline{O} \setminus S_\mathcal{I}}}
    \left[ \left| \sigma^{-1}_{\mathcal{I}}(i) \right| \cdot c(i, i^\prime) + f(i^\prime) \right],&
\end{cases}
\end{equation}
where $\min\emptyset:=+\infty$. Note that $\gamma(z,i)$ is finite since we have $|\SI|\geq 2$ or $|\overline{O} \setminus\SI|\geq 1$.
\end{itemize}
Finally, we choose $T=\Sinv$ if $|\Sinv|$ is even; we choose $T=\Sinv\cup\{z\}$ otherwise.
For notational convenience, for a set $E^\prime \subseteq E$, let $\gamma(E^\prime) := \sum_{e \in E^\prime} \gamma(e)$.

Intuitively speaking, if the $T$-join chooses a \emph{reassign edge} $(i,i')$, it is instructing us to reassign a client between $i$ and $i'$; an \emph{opening edge} $(z,i)$ corresponds to opening an initially closed facility~$i$; finally, a \emph{closing edge} $(z,i)$ corresponds to closing down an initially open facility $i$. Although we will formally describe this correction procedure later, here we introduce one more definition: if we decide to close down a facility $i$, we will need to reassign all clients that were previously assigned to $i$ to some other facility. This facility is called the \emph{substitute} of $i$. The substitutes are selected as the facilities that attain the minimum in \eqref{e:closingcost}. That is, for each closing edge $(z,i)$, the substitute of $i$, denoted by $\phi(i)$, is given as follows (ties are broken arbitrarily when the $\argmin$ has more than one element):\begin{equation}\label{e:phi}
\phi( i) \in \begin{cases}
\displaystyle \argmin_{i^\prime \in S_\mathcal{I} \setminus \{i\}}
    \left[ \left| \sigma^{-1}_{\mathcal{I}}(i) \right| \cdot c(i, i^\prime) \right]
    ,&\textrm{if }\displaystyle \gamma(z,i)=\min_{i^\prime \in S_\mathcal{I} \setminus \{i\}}
    \left[ \left| \sigma^{-1}_{\mathcal{I}}(i) \right| \cdot c(i, i^\prime) \right],\\[1em]
    \displaystyle \argmin_{i^\prime \in {\overline{O} \setminus S_\mathcal{I}}}
    \left[ \left| \sigma^{-1}_{\mathcal{I}}(i) \right| \cdot c(i, i^\prime) + f(i^\prime) \right]
    ,&\textrm{otherwise}.
\end{cases}
\end{equation}

\paragraph{Sparsifying the $T$-join.}
Given a minimum $T$-join $J$, we examine whether any of the following operations can be performed; if so, we perform the operation and repeat. We terminate when none of the operations can be applied any more.
\begin{enumerate}[(i)]\setlength{\itemsep}{0pt}
    \item\label{pp:1} If $(i,i_1),(i, i_2)\in J$ for some $i,i_1,i_2\in F$, remove $(i,i_1)$ and $(i, i_2)$ from $J$ and add $(i_1,i_2)$ instead.
    \item\label{pp:2} If $(z,i)\in J$ for some open odd-constrained facility $i$ and $(z,\phi(i))\in J$, remove $(z,i)$ and $(z,\phi(i))$ from $J$ and add $(i,\phi(i))$ instead. (Note that the condition implies that $(z,i)$ is a closing edge and $\phi(i)$ is, therefore, well-defined.)
    \item\label{pp:3} If $J$ contains a cycle, remove all edges on the cycle.
\end{enumerate}

\paragraph{Parity correction.}
The final step of the algorithm is to modify the initial solution as prescribed by the sparsified $T$-join $J$. The parity correction is performed in the following three substeps:
\begin{enumerate}\setlength{\itemsep}{0pt}
\item Firstly, for each opening edge $(z,i)\in J$, open $i$ and remove $(z,i)$ from $J$.
\item Secondly, for each reassign edge $(i_1,i_2)\in J$, reassign one arbitrary client from one of the two facilities to the other and remove $(i_1,i_2)$ from $J$ as follows:
\begin{itemize}\setlength{\itemsep}{0pt}
\item if $(z,i_1)\in J$ (or $(z,i_2)\in J$), reassign \emph{from} $i_1$ \emph{to} $i_2$ (or from $i_2$ to $i_1$, respectively);
\item otherwise, at least one of these facilities is guaranteed to be currently assigned at least one client; reassign from that facility to the other.
\end{itemize}
\item Lastly, for each closing edge $(z,i)\in J$, close $i$ and reassign all clients currently assigned to $i$ to $\phi(i)$; if necessary, open $\phi(i)$. Remove $(z,i)$ from $J$.
\end{enumerate}

\subsection{Analysis of the sparsification and parity correction}

In this section, we show that a sparsified $T$-join prescribes a cheap modification for correcting the invalid facilities in the initial solution.
We first prove that the sparsification does not increase the cost of a $T$-join.
We will slightly abuse the notation and treat a $T$-join $J$ interchangeably as a graph $(F\cup\{z\},J)$.
For $x\in F\cup\{z\}$, let $\dJ(x)$ denote the degree of $x$ in such a graph $J$.

\begin{lemma}
The given sparsification procedure yields a $T$-join of no greater cost.
\end{lemma}
\begin{proof}
It suffices to prove that each single operation produces a $T$-join of no greater cost, and the lemma follows from the induction on the number of operations. Observe that, for every vertex $x\in F\cup\{z\}$, the parity of $\dJ(x)$  remains the same when we apply any of the three operations. It remains to show that all three operations never increase the cost of $J$.

Consider Operation~\eqref{pp:1} that replaces $(i,i_1)$ and $(i,i_2)$ with $(i_1,i_2)$. The net increase in the cost is
\[
\gamma(i_1,i_2)-[\gamma(i, i_1)+\gamma(i,i_2)]=c(i_1,i_2)-[c(i, i_1)+c(i,i_2)]\leq 0
,\]where the inequality follows from the triangle inequality.

Operation~\eqref{pp:3} does not increase the cost of $J$ since all costs are nonnegative.

Now consider Operation~\eqref{pp:2}.
We can assume without loss of generality that $\mathcal{A}_\mathsf{FL}$ returns a solution such that $i\in\SI$ implies $\sigma^{-1}_{\mathcal{I}}(i)\neq\emptyset$. (Otherwise, we can simply exclude $i$ from \SI.) Since $(z,i)$ is a closing edge, \eqref{e:closingcost} and \eqref{e:phi} imply that $\gamma(z,i)\geq \left| \sigma^{-1}_{\mathcal{I}}(i) \right| \cdot c(i, \phi(i))\geq c(i, \phi(i)) = \gamma(i, \phi(i))$. The operation does not increase the cost of $J$ since $\gamma(z,\phi(i))\geq 0$.
\end{proof}

Following are the key sparsity observations we will use in the parity correction step. Let $J$ denote the sparsified $T$-join on which no further operation was possible.
\begin{observation}\label{obs:pre1}
For all $i\in F$, $i$ is adjacent in $J$ with at most one vertex in $F$.
\end{observation}
\begin{proof}
If $i$ were adjacent in $J$ with two facilities $i_1$ and $i_2$, $(i,i_1)$ and $(i,i_2)$ would have been replaced with $(i_1,i_2)$.
\end{proof}

\begin{observation}\label{obs:pre2}
For all edges $(i_1,i_2)\in J$ such that $i_1,i_2\in F$, at least one of $i_1$ and $i_2$ belongs to \SI, the set of initially open facilities.
\end{observation}
\begin{proof}
Suppose towards contradiction that $i_1,i_2\in F\setminus \SI$. Since $i_1\notin T$, there exists some vertex $x$ other than $i_2$ that is adjacent with $i_1$; we have $x=z$ from Observation~\ref{obs:pre1}. Likewise, we have $(z,i_2)\in J$, leading to contradiction since $\{(i_1,i_2),(z,i_1),(z,i_2)\}$ forms a cycle.
\end{proof}

\begin{observation}\label{obs:pre3}
If a closing edge $(z,i)$ is in $J$, we have $(z,\phi(i))\notin J$.
\end{observation}
\begin{proof}
Since $(z,i)\in J$ is a closing edge, we know $i$ is an open odd-constrained facility. Suppose $(z,\phi(i))\in J$. Then $(z,i)$ and $(z,\phi(i))$ would have been replaced with $(i,\phi(i))$, leading to contradiction.
\end{proof}
\newcommand{\ourSol}{\mathsf{ALG}}
\newcommand{\initSol}{\mathcal{I}}

We can now analyze the parity correction prescribed by $J$. Observations~\ref{obs:pre1} and~\ref{obs:pre2} show that, when we process a reassign edge, the facility that \emph{gives} a client has at least one client assigned to it, and the facility that \emph{receives} a client is open. Observation~\ref{obs:pre3} proves that, when we process a closing edge $(z,i)$, the substitute $\phi(i)$ is indeed open. These arguments are formalized by the following lemma.

\begin{lemma} \label{lem:map1}
The corrected solution is a feasible solution.
Moreover, the correction cost is bounded by $\gamma(J)$ from above.
\end{lemma}
\noindent Its full proof is deferred to Section~\ref{sec:dproofs}.

\subsection{Bounding \texorpdfstring{$\gamma(J)$}{Lg}}
We show in this section that the cost of a minimum $T$-join in the auxiliary graph $G$ is within a constant factor of the optimum.
Here we fix an arbitrary optimal solution $\SO\subseteq F$ and $\sO:D\to \SO$; let $\mathsf{OPT}:=f(\SO)+c(\sO)$ denote its value. In the rest of this section, we will exhibit a $T$-join dominator $Y \subseteq E$ such that $\gamma(Y) \leq O(1) \cdot \mathsf{OPT}$.

We construct $Y$ as the union of three edge sets $Y_1$, $Y_2$, and $Y_3$, defined as follows. (See Figure \ref{fig:tjd} on the last page.)
\begin{itemize}\setlength{\itemsep}{0pt}
    \item For each client $j\in D$, we add an edge to $Y_1$ between the facility $j$ is assigned to in the initial solution and the facility $j$ is assigned to in the fixed optimal solution.
If both facilities are the same, we just ignore the client rather than creating a loop.
    
    \item Let $Y_2$ be the set of edges between $z$ and each odd-constrained facility that is closed in the initial solution but open in the optimal solution.
    \item 
    Let $\mathcal{C}$ be the set of connected components $C$ in $(F \cup \{z\}, Y_1 \cup Y_2)$ such that $C$ does not contain $z$ and is $T$-odd, i.e.,
    $
    \mathcal{C} := \{ C  \mid  C \text{ is a connected component in } (F \cup \{z\}, Y_1 \cup Y_2),\; z \notin V(C) \text{,}
    $
    $
    \text{and } |V(C) \cap T| \text{ is odd} \},
    $ where $V(C)$ denotes the set of vertices in $C$.
    For each component $C \in \mathcal{C}$, pick an arbitrary odd-constrained facility $i^C\in V(C)$ that is open in the initial solution but closed in the optimal solution. We thus have $i^C\in V(C)\cap O\cap \SI\setminus\SO$.
    We now define $Y_3$ as the set of edges between $z$ and $i^C$ for all $C \in \mathcal{C}$.
\end{itemize}
The following observation holds since any two facilities that share a client have an edge in between in $Y_1$.
\begin{obs}\label{obeq}
For each $C \in \mathcal{C}$, let $n_\mathcal{I}^C$ (and $n_\mathcal{O}^C$) be the number of clients assigned to a facility in $V(C)$ by the initial solution (and by the optimal solution, respectively). We then have $
n_\mathcal{I}^C = n_\mathcal{O}^C
$.
\end{obs}

In order to show that $Y$ is a $T$-join dominator, it suffices to prove that $i^C$ can be chosen for every connected component $C \in \mathcal{C}$. Suppose that $i^C$ is well-defined for each $C\in\mathcal{C}$. Then, in the graph $(F\cup\{z\},Y_1\cup Y_2\cup Y_3)$, every connected component that does not contain $z$ must contain an even number of vertices in $T$, since otherwise, an edge in $Y_3$ would have connected this component to $z$. That is, the only connected component in $(F\cup\{z\},Y_1\cup Y_2\cup Y_3)$ that may have an odd number of vertices in $T$ is the one that contains $z$; however, since $|T|$ is even, this component also has even number of vertices in $T$. The conclusion now follows from the definition of a $T$-join dominator.

\begin{lemma} \label{existence-of-closed-odd-fac-in-T-odd}
For each connected component $C \in \mathcal{C}$, there exists an odd-constrained facility that is open in the initial solution but closed in the fixed optimal solution.
\end{lemma}
\begin{proof}
From construction, for any facility $i$ that is closed in both solutions, $i$ cannot be incident with any edges in $Y_1$ or $Y_2$. Such facility $i$, therefore, forms a singleton connected component in $(F\cup\{z\},Y_1\cup Y_2)$ and we have $i\notin T$. Thus, for the $T$-odd component $C$, we have that every facility in $V(C)$ must be open in at least one of the two solutions.

Suppose towards contradiction that there does not exist an odd-constrained facility that is open in the initial solution but closed in the fixed optimal solution. That is, every facility $i\in V(C)$ that is open only in the initial solution is even-constrained.

If some facility $i\in V(C)$ is open only in the \emph{optimal} solution, $i$ cannot be odd-constrained: otherwise, the opening edge $(z,i)$ would be in $Y_2$, contradicting $C\in\mathcal{C}$. So we now have that every facility $i\in V(C)$ that is open only in one of the two solutions is even-constrained. In other words, every facility in $V(C)$ is either even-constrained or open in both solutions. This, together with $n_\mathcal{O}^C=\sum_{i\in V(C)}|\sO^{-1}(i)|$ and the fact that the optimal solution satisfies all parity constraints, implies that the parity of $n_\mathcal{O}^C$ is the same as that of $|V(C)\cap O|$.

On the other hand, since $V(C)$ contains invalid facilities, we have that the parity of $n^C_\mathcal{I}$ is equal to that of $|V(C)\cap O|+|V(C)\cap \Sinv|$. From Observation~\ref{obeq}, this implies that $|V(C)\cap \Sinv|$ is even, contradicting the fact that $C$ is $T$-odd.
\end{proof}

To argue that the closing edge $(z,i^C)$ indeed exists in $G$, we need to verify that $|\SI|\geq 2$ or $|\Obar\setminus\SI|\geq 1$. Note that $|\SI|\geq 1$ as long as $D\neq\emptyset$.
\begin{lemma} \label{script-C-safe}
If $\left| S_\mathcal{I} \right| = 1$ and $\overline{O} \setminus S_\mathcal{I} = \emptyset$, we have $\mathcal{C} = \emptyset$.
\end{lemma}
\begin{proof}
Suppose towards contradiction that $\left| S_\mathcal{I} \right| = 1$, $\overline{O} \setminus S_\mathcal{I} = \emptyset$, and $C\in\mathcal{C}$. Since $|V(C)\cap T|$ is odd, we can choose some $i\in V(C)\cap T$. Moreover, we have $i\in T\setminus\{z\}=\Sinv\subseteq\SI$ since $z\notin V(C)$.

Now we claim that $V(C)=\{i\}$. (\emph{Proof.} Let $i'$ be an arbitrary facility in $V(C)\setminus\{i\}$. Since $\SI=\{i\}$, we have $i'\notin\SI$. This, together with $\overline{O} \setminus S_\mathcal{I} = \emptyset$, yields $i'\in O$. Since $C$ does not contain $z$, we have $(z,i')\notin Y_2$, showing $i'\notin\SO$. Recall from the proof of Lemma~\ref{existence-of-closed-odd-fac-in-T-odd} that a facility in $V(C)$ cannot be closed in both solutions.)
Thus, $i$ must be open in the optimal solution, too.
Since $V(C)$ does not contain any facility that is closed in the optimal solution, we cannot choose $i^C$, contradicting Lemma~\ref{existence-of-closed-odd-fac-in-T-odd}.
\end{proof}

Now we bound the cost of $Y$.

\begin{obs}\label{lem:btj:1}
We have $\gamma(Y_1) \leq c(\sigma_\mathcal{I}) + c(\sigma_\mathcal{O})$.
\end{obs}
\begin{proof}
By the definition of $Y_1$, for every edge $(i, i^\prime) \in Y_1$, there exists a client $j$ such that $\sigma_\mathcal{I}(j) = i$ and $\sigma_\mathcal{O}(j) = i^\prime$, yielding that
$\gamma(i, i^\prime) = c(i, i^\prime) \leq c(i, j) + c(i^\prime, j)$.
We thus have
\begin{align*}
    \gamma(Y_1)
    = \sum_{(i,i^\prime) \in Y_1} \gamma(i, i^\prime) 
    \leq \sum_{\substack{i \in F, \\ j \in \sigma^{-1}_\mathcal{I}(i)}} c(i,j)
        + \sum_{\substack{i^\prime \in F, \\ j \in \sigma^{-1}_\mathcal{O}(i^\prime)}} c(i^\prime,j) 
    = c(\sigma_\mathcal{I}) + c(\sigma_\mathcal{O}).
\end{align*}
\end{proof}

\begin{obs}\label{lem:btj:2}
We have $\gamma(Y_2) \leq f( S_\mathcal{O} )$.
\end{obs}
\begin{proof}
Note that the cost of each edge $(z, i) \in Y_2$ is the opening cost of $i$ and we add $(z, i)$ to $Y_2$ only if $i \in S_\mathcal{O} \setminus S_\mathcal{I}$.
\end{proof}

\begin{lemma}\label{lem:btj:3}
We have $\gamma(Y_3) \leq c(\sigma_\mathcal{I}) + c(\sigma_\mathcal{O}) + f(S_\mathcal{O})$.
\end{lemma}
\begin{proof}
Let $C$ be an arbitrary connected component in $\mathcal{C}$.
By Lemma \ref{existence-of-closed-odd-fac-in-T-odd}, $i^C$ is well-defined.
Let $B^C \subseteq D$ be the set of the clients that are assigned to $i^C$ in the initial solution, i.e., $B^C := \sigma^{-1}_\mathcal{I}(i^C)$.

We claim that $B^C$ can be partitioned into two sets $B^C_\mathsf{open}$ and $B^C_\mathsf{ec}$ where the former is a set of clients assigned in the optimal solution to a facility which is also open in the initial solution and the latter is to a facility which is even-constrained and closed in the initial solution.
For each $j\in B^C$, consider $\sO(j)=:i'\in\SO$.
Note that $i'\in V(C)$ from the construction of $Y_1$.
If $i'$ is open in the initial solution, since $i^C$ is closed in the optimal solution, we have $i'\in\SI\setminus\{i^C\}$.
Otherwise, since $i'$ is open in the optimal solution but $(z,i')$ was not chosen in $Y_2$, it must be the case that $i'$ is even-constrained.
Thus, to reiterate,
\begin{align*}
    B^C_\mathsf{open} &:= \left\{ j \;|\; \sigma_\mathcal{I}(j) = i^C \text{ and } \sigma_\mathcal{O}(j) \in (S_\mathcal{I} \setminus \{i^C \}) \cap V(C) \right\}; \\
    B^C_\mathsf{ec} &:= \left\{ j \;|\; \sigma_\mathcal{I}(j) = i^C \text{ and } \sigma_\mathcal{O}(j) \in (\overline{O} \setminus S_\mathcal{I}) \cap V(C) \right\}.
\end{align*}
Let $\lambda:=\left| B^C_\mathsf{open} \right| / \left| B^C \right|$.
Since $B^C_\mathsf{open}$ and $B^C_\mathsf{ec}$ forms a partition of $B^C$, we have
$0 \leq \lambda \leq 1$ and
$ \left| B^C_\mathsf{ec} \right| / \left| B^C \right| = 1 - \lambda $.

Now we bound $\gamma(z, i^C)$ from above using the assignment costs (in both solutions) of the clients in $B^C$, along with the opening costs of $V(C) \cap \SO$.

Assume for now that $\lambda >0$. Then we have
\begin{align}
    c(\sigma_\mathcal{I}|_{B^C_\mathsf{open}}) + c(\sigma_\mathcal{O}|_{B^C_\mathsf{open}})
    &= \sum_{j \in B^C_\mathsf{open}} c(i^C,j)
        + \sum_{j \in B^C_\mathsf{open}} c \left( \sigma_\mathcal{O}(j), j \right)
    \geq \sum_{j \in B^C_\mathsf{open}} c \left( i^C, \sigma_\mathcal{O} (j) \right) \nonumber\\
    &\geq \left| B^C_\mathsf{open} \right| \cdot
        \min_{i^\prime \in S_\mathcal{I} \setminus \{i^C\} } c(i^C, i^\prime)
     = \lambda \cdot \min_{i^\prime \in S_\mathcal{I} \setminus \{i^C\}}
        \left[ \left| \sigma^{-1}_\mathcal{I} (i^C) \right| \cdot c(i^C, i^\prime) \right], \label{D1-lowerbound} 
\end{align}
where the first inequality holds due to the triangle inequality and the second inequality is from the fact that, for every $j \in B^C_\mathsf{open}$, we have $\sigma_\mathcal{O}(j) \in S_\mathcal{I} \setminus \{i^C\}$. Note that the above trivially holds if $\lambda=0$.

Now suppose for the moment that $\lambda < 1$. Let $S^C_\mathsf{ec}$ be the set of the facilities that are assigned a client from $B^C_\mathsf{ec}$ in the optimal solution, i.e.,
$
S^C_\mathsf{ec} := \left\{ \sigma_\mathcal{O}(j) \mid j \in B^C_\mathsf{ec}  \right\}.
$
Let $i_\mathsf{ec}$ be the closest facility from $i^C$ in $S^C_\mathsf{ec}$.
We then have the following:
\begin{align}
    &c(\sigma_\mathcal{I}|_{B^C_\mathsf{ec}}) + c(\sigma_\mathcal{O}|_{B^C_\mathsf{ec}}) + f(S^C_\mathsf{ec})
    = \sum_{j \in B^C_\mathsf{ec}} c(i^C,j)
        + \sum_{j \in B^C_\mathsf{ec}} c\left(\sigma_\mathcal{O}(j), j\right)
        + \sum_{i \in S^C_\mathsf{ec}} f(i) \nonumber\\
    & \quad \geq \sum_{j \in B^C_\mathsf{ec}} c \left( i^C, \sigma_\mathcal{O}(j) \right)
        + \sum_{i \in S^C_\mathsf{ec}} f(i) 
    \geq \left| B^C_\mathsf{ec} \right|  \cdot c(i^C, i_\mathsf{ec})
        + f(i_\mathsf{ec}) \nonumber \\
    & \quad \geq (1 - \lambda) \cdot \left[ \left| \sigma^{-1}_\mathcal{I}(i^C) \right| \cdot c(i^C, i_\mathsf{ec})
        + f(i_\mathsf{ec}) \right] 
    \geq (1 - \lambda) \cdot \min_{i^\prime \in \overline{O} \setminus S_\mathcal{I}}
        \left[ \left| \sigma^{-1}_\mathcal{I}(i^C) \right| \cdot c(i^C, i^\prime) + f(i^\prime)  \right] \label{D2-lowerbound}
.\end{align}
Again, the above trivially holds when $\lambda=1$. Combining \eqref{D1-lowerbound} and \eqref{D2-lowerbound} yields
\begin{equation*}
    c(\sigma_\mathcal{I} |_{B^C}) + c(\sigma_\mathcal{O} |_{B^C}) + f(S^C_\mathsf{ec}) \geq \gamma(z, i^C).
\end{equation*}

It is noteworthy that $B^C$'s for $C \in \mathcal{C}$ are mutually disjoint since a client can be assigned to exactly one facility.
We can also observe that $S^C_\mathsf{ec}$'s are mutually disjoint because each facility belongs to at most one connected component in $\mathcal{C}$. We thus have
\begin{align*}
    \gamma(Y_3)
    = \sum_{C \in \mathcal{C}} \gamma(z, i^C) 
    \leq \sum_{C \in \mathcal{C}}
        \left(
            c(\sigma_\mathcal{I} |_{B^C})
            + c(\sigma_\mathcal{O} |_{B^C})
            + f(S^C_\mathsf{ec})
        \right) 
    \leq c(\sigma_\mathcal{I}) + c(\sigma_\mathcal{O}) + f(S_\mathcal{O}).
\end{align*}
\end{proof}

\begin{lemma} \label{lem:btj4}
There exists a $T$-join $J$ in $G$ whose cost is no more than
$2 \cdot \left( c(\sigma_\mathcal{I}) + c(\sigma_\mathcal{O}) + f(S_\mathcal{O}) \right)$.
\end{lemma}
\begin{proof}
Immediate from Lemmas~\ref{existence-of-closed-odd-fac-in-T-odd} and \ref{lem:btj:3}, and Observations \ref{lem:btj:1} and \ref{lem:btj:2}.

\end{proof}

We can now prove our main theorem.

\begingroup
\def\thetheorem{\ref{thm:general}}
\begin{theorem}
There exists an $O(1)$-approximation algorithm for the $O$-facility location problem.
\end{theorem}
\addtocounter{theorem}{-1}
\endgroup
\begin{proof}
Immediate from Lemmas \ref{lem:map1} and \ref{lem:btj4}.
Note that the approximation ratio of our algorithm is $3 \rho_\mathsf{FL} + 2$ since our algorithm returns a feasible solution of cost at most
\begin{equation*}
\left( c(\sigma_\mathcal{I}) + f(S_\mathcal{I}) \right) + 2 \cdot \left( c(\sigma_\mathcal{I}) + c(\sigma_\mathcal{O}) + f(S_\mathcal{O}) \right)
\leq (3 \rho_\mathsf{FL} + 2) \cdot \mathsf{OPT},
\end{equation*}
since $c(\sI)+f(\SI)\leq\rho_\mathsf{FL}\cdot \mathsf{OPT}$.
It can be easily verified that the algorithm runs in polynomial time.
\end{proof}

\subsection{Deferred proofs}\label{sec:dproofs}
This section presents the deferred proofs.

\begingroup
\def\thetheorem{\ref{obs:equiv}}
\begin{observation}
The two problem definitions are equivalent.
\end{observation}
\addtocounter{theorem}{-1}
\endgroup
\begin{proof}
Given an $O$-facility location problem instance, simply by defining $\pi$ as\[
\pi(i):=\begin{cases}\mathsf{odd},&\textrm{if }i\in O,\\\mathsf{even},&\textrm{otherwise},\end{cases}
\]we arrive at an equivalent instance of the second form, i.e., the form where unconstrained facilities are allowed. Now suppose that we are given an instance of the second form. We create two copies of every facility $i$ such that $\pi(i)=\mathsf{unconstrained}$ and put exactly one of these two copies into $O$, in addition to the facilities $i$ with $\pi(i)=\mathsf{odd}$. Since all opening costs are nonnegative, we can assume without loss of generality that an optimal solution to the new instance will open at most one copy of the duplicates. This shows the equivalence of the two problem definitions.
\end{proof}

\begingroup
\def\thetheorem{\ref{lem:map1}}
\begin{lemma}
The corrected solution is a feasible solution.
Moreover, the correction cost is bounded by $\gamma(J)$ from above.
\end{lemma}
\addtocounter{theorem}{-1}
\endgroup
\begin{proof}
We say a facility $i$ has the \emph{incorrect} parity in a solution $(S, \sigma)$ if the parity constraint of the facility is violated in the ``current'' solution. (This definition differs from the \emph{invalid} facilities, which are \emph{fixed} as the facilities with the incorrect parities in the initial solution.)
We show the feasibility of the corrected solution by establishing invariants throughout the parity correction procedure.
We modify the procedure so that it now modifies $T$ in addition to $J$ in each iteration.
The invariants are the following: 
\begin{itemize}\setlength{\itemsep}{0pt}
    \item $J$ is a $T$-join in the auxiliary graph, and
    \item $T \setminus \{ z \}$ is exactly the set of facilities having the incorrect parities.
\end{itemize}

Observe that $|J|$ decreases by one in each iteration; the corrected solution is, therefore, feasible since the empty set of edges is an $\emptyset$-join.
The correction cost will be bounded by showing that the cost incurred in each iteration can be covered by the cost of the corresponding edge removed from $J$.
Recall that we start with the sparsified $T$-join $J$ where $T\setminus\{z\} = S_\mathsf{inv}$; it is clear that both invariants initially hold.

Now we start erasing the edges from $J$.
Here we remark that, given a $T$-join $J$ and an edge $(i_1, i_2) \in J$, $J \setminus \{(i_1, i_2)\}$ is a $T \triangle \{i_1, i_2\}$-join since the degrees of $i_1$ and $i_2$ decrease by one.

Let us consider the first substep.
For each opening edge $(z, i) \in J$, we open facility $i$ in the solution, remove $(z, i)$ from $J$, and update $T \gets T \triangle \{z, i\}$.
It can be easily seen that $J$ is still a $T$-join.
By the construction of the auxiliary graph, facility $i$ is closed in the initial solution and hence $i \notin T$ at the beginning of this iteration. 
Moreover, since $i$ is odd-constrained at the same time, after opening $i$, this facility enters the set of facilities having the incorrect parities, establishing the second invariant.
Observe that the cost for opening $i$ can be covered by $\gamma(z, i)$.

Next, for each reassign edge $(i_1, i_2) \in J$, we transfer a client $j$ from one of the two facilities to the other.
We then remove $(i_1, i_2) \in J$ and update $T \gets T \triangle \{i_1, i_2\}$.
By Observation \ref{obs:pre2}, at least one of $i_1$ and $i_2$ was open in $S_\initSol$.
Recall that we can assume without loss of generality that $\mathcal{A}_\mathsf{FL}$ returns a solution such that $i\in\SI$ implies $\sigma^{-1}_{\mathcal{I}}(i)\neq\emptyset$.
Thus, we can always choose a client $j$ to be reassigned (see Observation~\ref{obs:pre1}).
Since exactly one client is reassigned, the parities of $i_1$ and $i_2$ are flipped and the second invariant is maintained.
Note that the cost of reassigning $j$ from, say, $i_1$ to $i_2$ is $-c(i_1, j) + c(i_2, j)\leq c(i_1, i_2)= \gamma(i_1, i_2)$ from the triangle inequality.
If the reassignment is from $i_2$ to $i_1$, the symmetric argument holds.

Let us now consider the last substep where we handle closing edges.
For each closing edge $(z, i) \in J$, we close $i$ and reassign all the clients currently assigned to $i$ to its substitute $\phi(i)$.
We then remove $(z, i)$ from $J$ and update $T \gets T \triangle \{z, i\}$.

Consider the time point right before closing $i$.
We claim that the number of clients assigned to $i$ is even.
Due to the construction of the auxiliary graph, we have $i \in O$.
Since we have already processed (and removed) all reassign edges, $i$ is adjacent with only $z$ at the moment and thus is in $T$.
This, from the induction hypothesis, implies that $i$ has the incorrect parity, i.e., $i$ is assigned even number of clients.
Therefore, reassigning all the clients assigned to $i$ to $\phi(i)$ would not change the parity of $\phi(i)$.
(Note that, if $\phi(i)$ was closed at the beginning, then $\phi(i) \in \overline{O}\setminus \SI \subseteq \overline{O}$.)
With the fact that $i$ becomes closed at this iteration, this shows that both invariants hold.

We finally verify that the correction cost here is no more than $\gamma(z, i)$.
Recall that we close facility $i$ and reassign every client $j$ assigned to $i$ to $\phi(i)$.
Thus, the change of the assignment cost for each $j$ is exactly $-c(i,j) + c(\phi(i), j)$.
As argued above, by the triangle inequality, we know that this value can be bounded by $c(i, \phi(i))$ from above.
We can thus see that the total assignment cost may increase by at most $|\sigma^{-1}_\initSol(i)| \cdot c(i, \phi(i))$ since the number of clients assigned to $i$ does not increase during the previous substeps.
(During the second substep, when we process $(i_1,i_2)$ and find that one of the facilities, say $i_1$, is adjacent with $z$ in $J$, we reassigned a client \emph{from} $i_1$ \emph{to} $i_2$.)
If $\phi(i)$ was closed at the beginning of the correction, we may need to open $\phi(i)$, but $\gamma(z, i)$ already pays for it.
If $\phi(i)$ was open, by Observation \ref{obs:pre3}, we know $\phi(i)$ will never be closed.
These together imply that the correction cost is no greater than $\gamma(z,i)$.
\end{proof}

\bibliographystyle{abbrv}
\bibliography{lit}

\appendix

\section{Analysis for the All-Even Case}\label{app:even}

We present the full analysis of our algorithm described in Section~\ref{sec:alleven}.
We assume that $|D|$ is even; otherwise, the instance is infeasible.
It is easy to see that our algorithm returns a feasible solution since every facility is assigned exactly twice the number of the clients it is assigned in $\sigma_\UncSol$.
Fix an arbitrary optimal solution to the original problem, and let $S_\mathcal{O}\subseteq F$ and $\sigma_\mathcal{O}:D\to S_\mathcal{O}$ denote this solution.

\begin{lemma} \label{bound-matching}
There exists a matching $M$ whose cost is no greater than $c(\sigma_\mathcal{O})$.
\end{lemma}
\begin{proof}
Let $i$ be a facility in $S_\mathcal{O}$. Observe that $i$ is assigned an even number of clients in the optimal solution: $|\sigma^{-1}_\mathcal{O}(i)|$ is even.
We can thus find a matching $M_i$ on $\sigma^{-1}_\mathcal{O}(i)$ by arbitrarily pairing them, and the cost of $M_i$ is at most $\sum_{j \in \sigma^{-1}_\mathcal{O}(i)} c(i,j)$ since, for every $(j_1, j_2) \in M_i$, $c(j_1, j_2) \leq c(i, j_1) + c(i, j_2)$.
Choose $M$ as the union of $M_i$ for all $i\in S_\mathcal{O}$. The lemma now follows from the fact that $\{\sigma^{-1}_\mathcal{O}(i)\}_{i\in S_\mathcal{O} }$ form a partition of $D$.
\end{proof}

\begin{lemma} \label{expect-bound}
$\mathbb{E}\left[ c(\sigma_\UncSol) + f(S_\UncSol) \right]
\leq \rho_\mathsf{FL} \cdot \left[ \frac{c (\sigma_\mathcal{O})}{2} + f(S_\mathcal{O}) \right].$
\end{lemma}
\begin{proof}
Observe that $(S_\mathcal{O}, \sigma_\mathcal{O} |_{D^\prime})$ is a feasible solution to the unconstrained classic facility location instance. Let  $(S_\mathcal{O^\prime}, \sigma_\mathcal{O^\prime} )$ be an optimal solution to this instance, and we have \[
\mathbb{E}\left[ c(\sigma_\mathcal{O^\prime}) + f(S_\mathcal{O^\prime})\right]
    \leq \mathbb{E}\left[c(\sigma_\mathcal{O}|_{D^\prime}) + f(S_\mathcal{O}) \right]
    .\]
    
Since we constructed a perfect matching and chose one of the two endpoints of each edge in the matching uniformly at random, the marginal probability that a client is in $D^\prime$ is exactly $\frac{1}{2}$; thus we have
\begin{align*}
     \mathbb{E}\left[c(\sigma_\mathcal{O}|_{D^\prime}) + f(S_\mathcal{O}) \right] 
    &= \sum_{j \in D} c(\sigma_\mathcal{O}(j),j)\cdot \Pr \left[ j \in D^\prime \right] + f(S_\mathcal{O}) \\
    &= \frac{c(\sigma_\mathcal{O})}{2} + f(S_\mathcal{O}),
\end{align*}
where the first line follows from the linearity of expectation. The desired conclusion follows from the fact that $\mathcal{A}_\mathsf{FL}$ is a $\rho_\mathsf{FL}$-approximation algorithm for the unconstrained facility location problem.
\end{proof}

\begin{lemma} \label{all-even-approx-ratio}
$\mathbb{E}\left[c(\sigma_\mathsf{ALG}) + f(S_\mathsf{ALG}) \right]
\leq (\rho_\mathsf{FL} + 1) \cdot c(\sigma_\mathcal{O}) + 2\rho_\mathsf{FL} \cdot f(S_\mathcal{O}).$
\end{lemma}
\begin{proof}
Observe that we have $c(\sigma_\mathsf{ALG}(\widehat{j_e}), \widehat{j_e}) \leq c(\sigma_\UncSol(j_e), j_e) + c(j_e, \widehat{j_e})$ for every client $\widehat{j_e} \in D \setminus D^\prime$ from the triangle inequality, yielding\[
    c(\sigma_\mathsf{ALG} |_{D \setminus D^\prime})
    \leq c(\sigma_\UncSol) + c(M^\star)
\]since $M^\star$ is a perfect matching. We thus have
\begin{align*}
    \mathbb{E} \left[ c(\sigma_\mathsf{ALG}) + f(S_\mathsf{ALG}) \right]
    &= \mathbb{E} \left[ c(\sigma_\mathsf{ALG} |_{D^\prime})
        + c(\sigma_\mathsf{ALG} |_{D \setminus D^\prime})
        + f(S_\mathsf{ALG}) \right] \\
    &\leq \mathbb{E} \left[  c(\sigma_\UncSol) +  c(\sigma_\UncSol) + c(M^\star) + f(S_\UncSol) \right]
        \\
    &\leq (\rho_\mathsf{FL} + 1) \cdot c(\sigma_\mathcal{O})
        + 2\rho_\mathsf{FL} \cdot f(S_\mathcal{O}),
\end{align*}where the last inequality follows from Lemmas~\ref{bound-matching} and~\ref{expect-bound}.
\end{proof}

\begin{thm}
There exists a randomized $2\rho_\mathsf{FL}$-approximation algorithm for the $\emptyset$-facility location problem.
\end{thm}
\begin{proof}
Immediate from Lemma~\ref{all-even-approx-ratio}. It is easy to observe that the algorithm runs in polynomial time.
\end{proof}

\section{Parity-Constrained $k$-Center}\label{app:kcenter}

\newcommand{\leqtau}{{\leq \tau}}

In this appendix, we present a $6$-approximation algorithm for the parity-constrained $k$-center problem.

\subsection{Preliminaries}

\paragraph*{Problem definition.}
In the parity-constrained $k$-center problem, we are given as the input a metric $c$ on a set of nodes $V$, an integer $k$, and parity constraints
$\pi : V \rightarrow \{ \oddconst, \evenconst ,\unconst \}$.
The objective is to find a subset $S \subseteq V$ of size at most $k$ and an assignment $\sigma:V\to S$ such that, for all $u\in S$, $|\sigma^{-1}(u)|$ is odd if $\pi(u)=\oddconst$ and even if $\pi(u)=\evenconst$ so as to minimize $\max_{v\in V}c(\sigma(v),v)$. 

To simplify the presentation, we will assume that $\pi(u)\neq\unconst$ for all $u\in V$ in what follows. We will discuss how we can lift this assumption towards the end of this appendix.

\paragraph{Notation.}
For a given $\tau\in\mathbb{R}$, let $G_{\leq\tau}$ be an unweighted graph on node set $V$, in which $u$ and $v$ are adjacent if and only if $c(u,v)\leq \tau$. Let $S^*\subseteq V$ and $\sigma^*:V\to S^*$ be an optimal solution of value $\tau^*$; observe that, for all $v\in V$, either $\sigma^*(v)=v$ or $(\sigma^*(v),v)\in G_{\leq\tau^*}$. That is, $G_{\leq\tau^*}$ can be intuitively understood as a graph showing ``admissible'' assignments.
 
Given an unweighted graph $G = (V, E)$ and two nodes $u, v \in V$, let $d_G(u,v)$ denote the shortest length of a path between $u$ and $v$. For a positive integer $\rho\in\mathbb{Z}_{>0}$, the \emph{$\rho$-th power graph} of $G$, denoted by $G^\rho$, is a graph on the same node set $V$ where two nodes $u$ and $v$ are adjacent if and only if $d_G(u,v) \leq \rho$.

\paragraph{Guessing the optimum.}
Using the above notation, we can describe the standard method to solve bottleneck optimization problems~\cite{hochbaum1985}: the algorithm guesses the optimal value $\tau$ and constructs $G_{\leq\tau}$. If we can find some $S\subseteq V$ and $\sigma:V\to S$ such that $|S|\leq k$ and $d_{G_{\leq\tau}}(\sigma(v),v)\leq\rho$ for all $v\in V$, $S$ is a $\rho$-approximate solution. (Note that $d_{G_{\leq\tau}}(\sigma(v),v)\leq\rho$ implies $c(\sigma(v),v)\leq\rho\tau$ from the triangle inequality.) On the other hand, if the algorithm correctly concludes that there is no way to choose $S\subseteq V$ with $|S|\leq k$ such that every node $v\in V\setminus S$ is adjacent to a center in $S$, the guess is incorrect. This ``guessing'' can be done, for example, by performing binary search.

We note that we can assume that $G_{\leq\tau}$ is connected in applying this method, as in~\cite{cygan2012,an2015}: since no assignments can be made across different connected components in an optimal solution, we can separately consider each connected component and determine the smallest $k'$ for which we can find an assignment. If the total sum of these $k'$s exceeds the given budget $k$, we conclude that the guess was incorrect; otherwise, we can output the union of the centers chosen in each connected component.

\paragraph*{An algorithm for the unconstrained version.} For the sake of completeness, we will first present a $2$-approximation algorithm for the unconstrained $k$-center problem. This $2$-approximation algorithm is slightly different from the original presentation of Hochbaum and Shmoys~\cite{hochbaum1985}, but is a mere combination of the ideas already existing in \cite{hochbaum1985} and Khuller \& Sussmann~\cite{khuller2000}. Note that the following lemma suffices to obtain a $2$-approximation algorithm: the rest of the argument follows from the standard method of guessing the optimum. Although Conditions~\eqref{c:a:3} and \eqref{c:a:4} are not necessary to prove the correctness of the algorithm, they will be useful in our algorithm for the parity-constrained problem.

\begin{lemma} \label{akci:lem1}
Given an unweighted connected graph $G = (V, E)$, we can find in polynomial time a set of \emph{centers} $S \subseteq V$, an \emph{assignment} $\sigma:V\to S$, and a tree  $\mathcal{T}$ on $S$ satisfying the following conditions:
\begin{enumerate}[(i)]
\item there does not exist a set $S'\subseteq V$ with $|S'|<|S|$ such that every node $v\in V\setminus S'$ is adjacent to a center in $S'$;\label{c:a:1}
\item for every node $v \in V$, we have $d_G(\sigma(v),v)\leq 2$;\label{c:a:2}
\item for every center $u \in S$, we have $\sigma(u)=u$;\label{c:a:3}
\item for every edge $(u, v)$ in $\mathcal{T}$, we have $d_G(u, v) = 3$.\label{c:a:4}
\end{enumerate}
\end{lemma}
\begin{proof}[Proof sketch~\cite{hochbaum1985,khuller2000}]
Consider the following algorithm that chooses $S$ and constructs $\mathcal{T}$ at the same time.
Initially, we choose an arbitrary node $v$ to form a singleton tree. At this point, $v$ will be the only center in $S$.
We then repeatedly select a node $v\in V\setminus S$ that is at distance 3 from the centers chosen so far: i.e., $\min_{u\in S} d_G(u,v)=3$; this node $v$ enters $S$ and becomes the child of a node $u$ in $\mathcal{T}$ such that $d_G(u,v)=3$.

When we cannot admit any more node to the tree, every node $v\in V\setminus S$ must be at distance at most two from $S$. We construct the assignment $\sigma$ by assigning every center in $S$ to itself and every other node to an arbitrary center in $S$ that is at distance two or shorter. Note that Condition~\eqref{c:a:1} follows from the observation that no two centers in $S$ can be adjacent to a same node.
\end{proof}

\subsection{Our algorithm}\label{app:b:alg}

Appendices~\ref{app:b:alg} and~\ref{app:b:analysis} are dedicated to proving Lemma~\ref{akc:lem1b}.
Recall that we assume $\pi(v)\in\{\oddconst,\evenconst\}$ for all $v\in V$.
The algorithm will be presented in this section first and the analysis will follow in Appendix~\ref{app:b:analysis}.

\begin{lemma}\label{akc:lem1b}
Given an unweighted connected graph $G = (V, E)$, we can in polynomial time either find a set of \emph{centers} $S \subseteq V$ and an \emph{assignment} $\sigma:V\to S$ satisfying the following conditions, or correctly conclude that no such $(S,\sigma)$ exists:
\begin{enumerate}[(i)]
\item for every node $v \in V$, we have $d_G(\sigma(v),v)\leq 6$;\label{c:b:1}
\item for every center $u \in S$, we have $|\sigma^{-1}(u)|$ is odd if and only if $\pi(u)=\oddconst$;\label{c:b:2}
\item there does not exist a set $S'\subseteq V$ with $|S'|<|S|$ such that an assignment $\sigma':V\to S'$ satisfying the following conditions exists:\label{c:b:3}
\begin{itemize}
\item for every node $v \in V$, we have $d_G(\sigma'(v),v)\leq 1$;\label{c:c:1}
\item for every node $u \in S'$, we have $|\sigma'^{-1}(u)|$ is odd if and only if $\pi(u)=\oddconst$.\label{c:c:2}
\end{itemize}
\end{enumerate}
\end{lemma}

We reiterate that the entire algorithm first guesses the optimal value $\tau$, constructs $G_{\leq\tau}$, and executes the algorithm of this section for each connected component of it. If the algorithm reports no $(S,\sigma)$ exists for any one of the components, the guess is incorrect; if the union of obtained $S$'s has more than $k$ centers, the guess again is incorrect; otherwise, the union is the desired solution. This is a standard technique of bottleneck optimization~\cite{hochbaum1985,khuller2000} and we omit the formal proof.

\paragraph{Phase 1: Infeasibility test and initialization.}
The first step of our algorithm is to check if $|V|$ is odd but all nodes are even-constrained; in this case,
the algorithm concludes that no $(S,\sigma)$ exists. Otherwise, we execute the algorithm in the proof of Lemma~\ref{akci:lem1} to obtain an initial $S\subseteq V$ and $\sigma:V\to S$ along with a tree $\mathcal{T}$ on $S$. The algorithm will modify this $S$ and $\sigma$ to satisfy the parity constraints.

We say a node is \emph{invalid} if its parity constraint is not satisfied. Let $\Sinv \subseteq S$ denote the set of \emph{invalid} centers.
We can assume 
that $\Sinv \neq \emptyset$ since otherwise, the initial $S$ and $\sigma$ can be immediately output by our algorithm. 
As the algorithm proceeds, \Sinv will change.

\paragraph{\emph{Small case} ($|S|=1$).}
We will separately handle the case where $|S|=1$.
In this case, if there exists a node $v \in V$ whose parity constraint coincides with the parity of $|V|$, we close the center in $S$, open $v$ instead, and reassign every node in $V$ to $v$.
Otherwise, it is guaranteed that there exists a closed odd-constrained node $v$. We open $v$ in addition to the already open center in $S$ and reassign $v$ to itself.

Assume from now on that $|S|\geq 2$.

\paragraph{Phase 2: Modifying and rooting $\mathcal{T}$.}
We modify $\mathcal{T}$ and transform it into a rooted tree according to one of the following three cases:
\begin{enumerate}[(A)]
\item If $|\Sinv|$ is even, we simply root $\mathcal{T}$ at an arbitrary open node.\label{case:a}
\item Otherwise, we find an open odd-constrained center and make it the root of $\mathcal{T}$.\label{case:b}
\item If neither of the above applies, it is assured that there exists a closed odd-constrained node $v$.\label{case:c}
We open the node in addition to those in $S$, and add the node to $\mathcal{T}$ by attaching it as a child of $\sigma(v)$.
We root the tree at an arbitrary invalid center (except for $v$).
\end{enumerate}

\paragraph{Phase 3: Opening and closing centers.}
Let $r$ be the center selected as the root of $\mathcal{T}$. In Cases~\eqref{case:b} and~\eqref{case:c}, we close $r$, choose an arbitrary child $c$ of $r$, and reassign to $c$ every node $r$ was assigned. Even after $r$ is closed, it remains to be the root node of $\mathcal{T}$.

Note that the above procedure may change $\Sinv$: if $r$ was in \Sinv, $r$ leaves \Sinv as it gets closed. If $r$ was originally assigned an odd number of nodes, the parity of the number of nodes assigned to $c$ becomes flipped, so $c$ leaves \Sinv if it already was there and enters \Sinv if it was not. (If $r$ was originally assigned an even number of nodes, $c$ will stay still.) In Case~\eqref{case:c}, $v$ enters \Sinv.

\paragraph{Phase 4: Identifying reassignment paths.}
Once we obtained the modified and now-rooted tree $\mathcal{T}$, we construct node-disjoint paths on $\mathcal{T}^2$ between invalid nodes, called the \emph{reassignment paths}.
Let $\mathcal{T}_u$ denote the subtree rooted at $u\in \mathcal{T}$.

Given $\mathcal{T}_u$, the following recursive procedure finds a set of node-disjoint paths such that there is a one-to-one correspondence between the invalid nodes in $\mathcal{T}_u$ and the endpoints of the paths, provided that $\mathcal{T}_u$ contains an even number of invalid nodes. Obviously, this is impossible to achieve if the subtree contains an odd number of invalid nodes. In this case, every path will be between two invalid nodes, except for one that is between $u$ and an invalid node. This last path, called the \emph{exposed path} from $\mathcal{T}_u$, is allowed to be a trivial $u-u$ path if $u$ is invalid itself. Every invalid node still appears as an endpoint of exactly one path. The recursive procedure is indeed simple:\begin{itemize}
\item For each child $c$ of $u$, we recursively find the set of paths within $\mathcal{T}_c$. If $\mathcal{T}_c$ contains an odd number of invalid nodes, one of the paths will be an exposed path.
\item Collect these exposed paths, pair them arbitrarily, and concatenate each pair into a path between two invalid nodes. (Note that siblings are adjacent in $\mathcal{T}^2$.)
\item This may leave at most one exposed path that was not paired with another. We extend the path to $u$ by adding an edge. If $u$ is invalid, this yields a path between two invalid nodes. Otherwise, this path becomes the exposed path from $\mathcal{T}_u$.
\end{itemize}
It is guaranteed that $|\Sinv|$ is even at the beginning of this phase. Thus, when we run this recursive procedure on the entire tree $\mathcal{T}$, we will obtain a set of paths between invalid nodes, without an exposed path from $\mathcal{T}_r$.

\paragraph{Phase 5: Reassignment.}
For each path $P= u_1, \cdots, u_\ell$ we found above, we reassign $u_i$ to $u_{i+1}$ for all $i=1,\cdots,\ell-1$. Each path can be ``oriented'' in an arbitrary direction with the only exception that, if $\mathcal{T}$ was modified according to Case~\ref{case:c}, the newly opened node $v$ must be the last node of a path, i.e., the path containing $v$ must be oriented \emph{towards} $v$.

\subsection{Analysis}\label{app:b:analysis}

There are a few places in the algorithm where certain assertions are made. We first need to verify these are indeed true.
\begin{lemma}
The given algorithm is well-defined.
\end{lemma}
\begin{proof}
First, in the \emph{small case} ($|S|=1$), if we cannot find a node $v$ whose parity constraint coincides with the parity of $|V|$, this must be because $|V|$ is even but all nodes are odd-constrained. (Note that, if $|V|$ is odd and all nodes are even-constrained, the algorithm would have already concluded that no $(S,\sigma)$ exists in Phase~1.) Since $|V|$ is even and $|S|=1$, there must be some node in $V\setminus S$ and it has to be odd-constrained. This verifies the guarantee in the algorithm description.

Second, Case~\ref{case:c} of Phase~2 applies when $|\Sinv|$ is odd and there does not exist an open odd-constrained center. Since every open center is even-constrained and $|\Sinv|$ is odd, we have that $|V|$ is odd. Given that the algorithm did not conclude that no $(S,\sigma)$ exists, there must exist an odd-constrained node, verifying the guarantee.

Last, we verify the claim in the algorithm description that $|\Sinv|$ becomes even after Phase~3. 
If Case~\ref{case:a} applied in Phase~2, \Sinv does not change during Phases~2 and~3 and there is nothing to prove.

Suppose that Case~\ref{case:b} applied. If the root $r$ was valid, it was assigned an odd number of nodes; thus, the child $c$ to which all the nodes assigned to $r$ are reassigned gets its parity flipped. This flips the parity of $\Sinv$, making it even. On the other hand, if $r$ was invalid, the parity of $c$ is not flipped but $r$ itself is removed from \Sinv. This makes $|\Sinv|$ even again.

Finally, suppose that Case~\ref{case:c} applied. Closing the root $r$ flips the parity of $|\Sinv|$. Since none of the open centers are odd-constrained, $r$ must be even-constrained; therefore, the child $c$ gets its parity flipped, too. The newly opened odd-constrained node $v$ enters \Sinv. In total, the parity of $|\Sinv|$ flips three times, making it even.
\end{proof}

We are now ready to prove our main lemma.

\begin{proof}[Proof of Lemma~\ref{akc:lem1b}]
When the algorithm concludes that no $(S,\sigma)$ exists, this is because $|V|$ is odd but all nodes are even-constrained. Hence, the algorithm's conclusion is correct.

Observe that \Sinv indeed is the set of invalid centers at the end of Phase~3. It is easy to show by induction that the set of paths we found in Phase~4 are node-disjoint and their endpoints have a one-to-one correspondence between \Sinv, in addition to the fact that the root node does not appear in any of these paths unless the root itself is invalid.

From the construction, the reassignment in Phase~5 assigns nodes only to open centers. Due to the node disjointness of the paths, it is also clear that the parity constraints are all satisfied by the end of the algorithm: when we perform a series of reassignments on $u_1, \cdots, u_\ell$, the number of nodes each $u_i$ is assigned remains the same, except for $u_1$ (which decreases by 1) and $u_\ell$ (increases by 1). This verifies Condition~\eqref{c:b:2} of Lemma~\ref{akc:lem1b}.

In order to verify Condition~\eqref{c:b:1}, we consider each reassignment performed by the algorithm. In the small case, let $s\in S$ be the initially chosen center; then, for every node $v\in V$, we have $d_G(s,v)\leq 2$. We thus have $d_G(u,v)\leq 4$ for all $u,v\in V$, showing that all reassignments are within the distance of four no matter what our eventual choice of centers is.

In Phase~2, no reassignments are made but $\mathcal{T}$ may be modified; note that, for every $u$ and $v$ that are adjacent in $\mathcal{T}$, we have $d_G(u,v)\leq 3$. (The inequality may strictly hold when one of them is the newly added node in Case~\ref{case:c}.)

In Phase~3, when we close $r$, every node $v$ such that $\sigma(v)=r$ will be reassigned to $c$. Since $c$ is a child of $r$, we have $d_G(c,v)\leq d_G(c,r)+d_G(r,v)\leq 3+2$.

Finally in Phase~4, each $u_i$ is reassigned to $u_{i+1}$; we have $d_G(u_{i+1},u_i)\leq 6$ since $u_i$ and $u_{i+1}$ are adjacent in $\mathcal{T}^2$.

Now we verify Condition~\eqref{c:b:3}. Our algorithm does not increase the cardinality of open centers, except when $|S|=1$ and there does not exist a node whose constraint coincides with the parity of $|V|$. In that case, it is obvious that we have to open at least two centers. In the other cases, Condition~\eqref{c:a:1} of Lemma~\ref{akci:lem1} gives the desired conclusion.
\end{proof}

\subsection{Final remarks}

\paragraph{Allowing unconstrained nodes.}
It is easy to handle unconstrained nodes: for each connected graph $G=(V,E)$, we treat every unconstrained node as if its parity constraint is the parity of $|V|$.

When we treat an unconstrained node as a node with a fixed parity constraint, we only need to verify that this does not lead the algorithm to incorrectly conclude that no $(S,\sigma)$ exists or to unnecessarily open more centers than the unconstrained algorithm does. Observe that both can happen only in the small case of our algorithm. However, since we treat an unconstrained node as having the parity of $|V|$, the algorithm will always succeed in opening exactly one center in the small case.

\paragraph*{}
This completes the proof of Theorem~\ref{thm:kcentermain}.

\begingroup
\def\thetheorem{\ref{thm:kcentermain}}
\begin{theorem}[rephrased]
There exists a 6-approximation algorithm for the parity-constrained $k$-center problem.
\end{theorem}
\addtocounter{theorem}{-1}
\endgroup

\paragraph{Running time.}
We can implement the algorithm in the proof of Lemma~\ref{akci:lem1} to run in $O(|V|^2)$ time as follows: we maintain $\min_{u\in S}d_G(u,v)$ for each node $v\in V$. This array can be initialized as $+\infty$ and updated by performing BFS each time a new node enters $S$. A na\"ive implementation, therefore, runs in $O(|V|^3)$ time.
However, we do not need to exactly determine the values in the array; if a value is greater than 3, it suffices to know that it is greater than 3.
Thus, we can initialize all the distances as 4 instead of $+\infty$, and when we run BFS, we can modify it so that it does not visit a node to which the shortest path length is already greater than 3. In this algorithm, the value of every node can only decrease and only then we will let BFS ``visit'' the node; this can happen at most four times for each node. The overall running time, therefore, is bounded by $O(|V|^2)$. Now the rest of the algorithm is a straightforward traversal of $\mathcal{T}$ with reassignments, which runs in $O(|V|)$ time.

The overall running time is $O(|V|^2\log |V|)$, where the $\log |V|$ term originates from binary search.

\begin{obs}
Our algorithm can be implemented to run in $O(|V|^2\log |V|)$ time.
\end{obs}

\captionsetup[figure]{font=small, labelfont=bf}
\begin{figure}[p]
\centering
\begin{tikzpicture}

\newcommand\nodeSize{0.4}
\newcommand\borderwidth{1.1mm}
\newcommand\borderwidthbold{2.0mm}
\newcommand\facXpos{{1, 5.5, 6, 5, 12.5, 11, 17}}
\newcommand\facYpos{{7, 0.5, 6, 11, 1, 6.5, 5}}
\newcommand\specialXpos{12}
\newcommand\specialYpos{11}
\newcommand\cliXpos{{2.5, 8.75, 4, 6.5, 8.5, 2.5, 6, 11, 13, 16, 14.5}}
\newcommand\cliYpos{{2, 1, 4, 3, 3.5, 9, 8.5, 3.75, 3.25, 2, 5}}

\drawSolAssign{\cliXpos[0]}{\cliYpos[0]}{\facXpos[1]}{\facYpos[1]}{\nodeSize}
\drawSolAssign{\cliXpos[1]}{\cliYpos[1]}{\facXpos[1]}{\facYpos[1]}{\nodeSize}
\drawSolAssign{\cliXpos[2]}{\cliYpos[2]}{\facXpos[2]}{\facYpos[2]}{\nodeSize}
\drawSolAssign{\cliXpos[3]}{\cliYpos[3]}{\facXpos[2]}{\facYpos[2]}{\nodeSize}
\drawSolAssign{\cliXpos[4]}{\cliYpos[4]}{\facXpos[2]}{\facYpos[2]}{\nodeSize}
\drawSolAssign{\cliXpos[5]}{\cliYpos[5]}{\facXpos[3]}{\facYpos[3]}{\nodeSize}
\drawSolAssign{\cliXpos[6]}{\cliYpos[6]}{\facXpos[3]}{\facYpos[3]}{\nodeSize}
\drawSolAssign{\cliXpos[7]}{\cliYpos[7]}{\facXpos[4]}{\facYpos[4]}{\nodeSize}
\drawSolAssign{\cliXpos[8]}{\cliYpos[8]}{\facXpos[4]}{\facYpos[4]}{\nodeSize}
\drawSolAssign{\cliXpos[9]}{\cliYpos[9]}{\facXpos[6]}{\facYpos[6]}{\nodeSize}
\drawSolAssign{\cliXpos[10]}{\cliYpos[10]}{\facXpos[6]}{\facYpos[6]}{\nodeSize}

\drawOptAssign{\cliXpos[0]}{\cliYpos[0]}{\facXpos[0]}{\facYpos[0]}{\nodeSize}
\drawOptAssignCurve{\cliXpos[1]}{\cliYpos[1]}{\facXpos[1]}{\facYpos[1]}{7.5}{0.25}{\nodeSize}
\drawOptAssignCurve{\cliXpos[2]}{\cliYpos[2]}{\facXpos[2]}{\facYpos[2]}{4.5}{5.5}{\nodeSize}
\drawOptAssign{\cliXpos[3]}{\cliYpos[3]}{\facXpos[1]}{\facYpos[1]}{\nodeSize}
\drawOptAssign{\cliXpos[4]}{\cliYpos[4]}{\facXpos[1]}{\facYpos[1]}{\nodeSize}
\drawOptAssign{\cliXpos[5]}{\cliYpos[5]}{\facXpos[0]}{\facYpos[0]}{\nodeSize}
\drawOptAssign{\cliXpos[6]}{\cliYpos[6]}{\facXpos[2]}{\facYpos[2]}{\nodeSize}
\drawOptAssign{\cliXpos[7]}{\cliYpos[7]}{\facXpos[5]}{\facYpos[5]}{\nodeSize}
\drawOptAssign{\cliXpos[8]}{\cliYpos[8]}{\facXpos[5]}{\facYpos[5]}{\nodeSize}
\drawOptAssign{\cliXpos[9]}{\cliYpos[9]}{\facXpos[4]}{\facYpos[4]}{\nodeSize}
\drawOptAssign{\cliXpos[10]}{\cliYpos[10]}{\facXpos[5]}{\facYpos[5]}{\nodeSize}

\drawEvenCO{\facXpos[0]}{\facYpos[0]}{\nodeSize}{\borderwidth}
\drawOddOO{\facXpos[1]}{\facYpos[1]}{\nodeSize}{\borderwidth}
\drawEvenOO{\facXpos[2]}{\facYpos[2]}{\nodeSize}{\borderwidth}
\drawOddOC{\facXpos[3]}{\facYpos[3]}{\nodeSize}{\borderwidth}
\drawOddOO{\facXpos[4]}{\facYpos[4]}{\nodeSize}{\borderwidth}
\drawOddCO{\facXpos[5]}{\facYpos[5]}{\nodeSize}{\borderwidth}
\drawEvenOC{\facXpos[6]}{\facYpos[6]}{\nodeSize}{\borderwidth}

\drawClient{\cliXpos[0]}{\cliYpos[0]}{\nodeSize}
\drawClient{\cliXpos[1]}{\cliYpos[1]}{\nodeSize}
\drawClient{\cliXpos[2]}{\cliYpos[2]}{\nodeSize}
\drawClient{\cliXpos[3]}{\cliYpos[3]}{\nodeSize}
\drawClient{\cliXpos[4]}{\cliYpos[4]}{\nodeSize}
\drawClient{\cliXpos[5]}{\cliYpos[5]}{\nodeSize}
\drawClient{\cliXpos[6]}{\cliYpos[6]}{\nodeSize}
\drawClient{\cliXpos[7]}{\cliYpos[7]}{\nodeSize}
\drawClient{\cliXpos[8]}{\cliYpos[8]}{\nodeSize}
\drawClient{\cliXpos[9]}{\cliYpos[9]}{\nodeSize}
\drawClient{\cliXpos[10]}{\cliYpos[10]}{\nodeSize}

\drawTextA{9}{-2}{\nodeSize}
\end{tikzpicture}
\hspace{1cm}
\begin{tikzpicture}

\newcommand\nodeSize{0.4}
\newcommand\borderwidth{1.1mm}
\newcommand\borderwidthbold{2.0mm}
\newcommand\facXpos{{1, 5.5, 6, 5, 12.5, 11, 17}}
\newcommand\facYpos{{7, 0.5, 6, 11, 1, 6.5, 5}}
\newcommand\specialXpos{12}
\newcommand\specialYpos{11}

\drawEdgeGrey{\facXpos[1]}{\facYpos[1]}{\specialXpos}{\specialYpos}{\nodeSize}
\drawEdgeGrey{\facXpos[3]}{\facYpos[3]}{\specialXpos}{\specialYpos}{\nodeSize}
\drawEdgeGreyCurve{\facXpos[4]}{\facYpos[4]}{\specialXpos}{\specialYpos}{13.5}{6}{\nodeSize}

\drawYa{\facXpos[0]}{\facYpos[0]}{\facXpos[1]}{\facYpos[1]}{\nodeSize}
\drawYa{\facXpos[0]}{\facYpos[0]}{\facXpos[3]}{\facYpos[3]}{\nodeSize}
\drawYa{\facXpos[1]}{\facYpos[1]}{\facXpos[2]}{\facYpos[2]}{\nodeSize}
\drawYa{\facXpos[2]}{\facYpos[2]}{\facXpos[3]}{\facYpos[3]}{\nodeSize}
\drawYa{\facXpos[4]}{\facYpos[4]}{\facXpos[5]}{\facYpos[5]}{\nodeSize}
\drawYa{\facXpos[4]}{\facYpos[4]}{\facXpos[6]}{\facYpos[6]}{\nodeSize}
\drawYa{\facXpos[5]}{\facYpos[5]}{\facXpos[6]}{\facYpos[6]}{\nodeSize}
\drawYb{\facXpos[5]}{\facYpos[5]}{\specialXpos}{\specialYpos}{\nodeSize}
\drawYc{\facXpos[3]}{\facYpos[3]}{\specialXpos}{\specialYpos}{\nodeSize}

\drawEvenCO{\facXpos[0]}{\facYpos[0]}{\nodeSize}{\borderwidth}
\drawOddOO{\facXpos[1]}{\facYpos[1]}{\nodeSize}{\borderwidthbold}
\drawEvenOO{\facXpos[2]}{\facYpos[2]}{\nodeSize}{\borderwidthbold}
\drawOddOC{\facXpos[3]}{\facYpos[3]}{\nodeSize}{\borderwidthbold}
\drawOddOO{\facXpos[4]}{\facYpos[4]}{\nodeSize}{\borderwidthbold}
\drawOddCO{\facXpos[5]}{\facYpos[5]}{\nodeSize}{\borderwidth}
\drawEvenOC{\facXpos[6]}{\facYpos[6]}{\nodeSize}{\borderwidth}
\drawSpecial{\specialXpos}{\specialYpos}{\nodeSize}

\drawTextB{9}{-2}{\nodeSize}
\end{tikzpicture}

\caption{\textbf{(A)}: An initial solution $(S_\mathcal{I}, \sigma_\mathcal{I})$ and an optimal solution $(S_\mathcal{O}, \sigma_\mathcal{O})$.
A facility (and a client) is represented as a square (and a circle, respectively).
Odd-constrained facilities have red dashed borders; even-constrained ones have navy-blue solid borders.
The upper-right triangle is filled with black if the facility is open in the initial solution; the lower-left triangle is filled with gray if it is open in the optimal solution.
Assignments in the initial solution are marked with black solid lines; assignments in the optimal solution are gray solid lines.
\textbf{(B)}: A $T$-join dominator $Y$.
If a facility $i$ is in $S_\mathsf{inv}$, it is marked with a thicker border.
Every edge in $Y_1$ is drawn as a green densely-dotted line; $Y_2$ as a navy-blue loosely-dotted line; and $Y_3$ as a coral solid line.
The remaining closing edges are drawn as black thin solid lines.
(We omitted the remaining reassign edges.)
}\label{fig:tjd}
\end{figure}
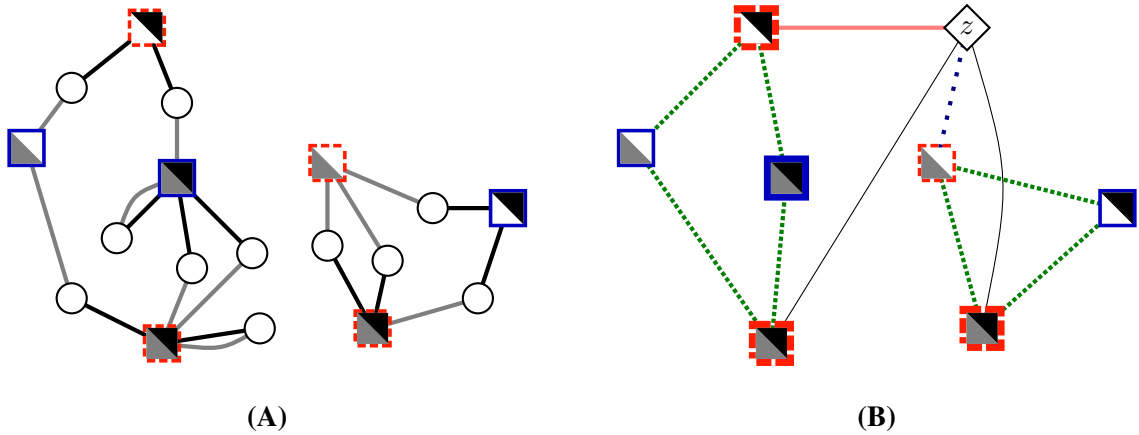

\end{document}